\documentclass[11pt]{amsart}
\usepackage{amssymb,mathrsfs,graphicx,enumerate}
\usepackage{amsmath,amsfonts,amssymb,amscd,amsthm,bbm}

\usepackage[retainorgcmds]{IEEEtrantools}
\usepackage{colortbl}
\usepackage{cite}
\usepackage{subcaption}
\usepackage{bm}
\title[]{Emergent behaviors of Cucker-Smale flocks on the hyperboloid}

\author[Ahn]{Hyunjin Ahn}
\address[Hyunjin Ahn]{\newline Department of Mathematical Sciences\newline Seoul National University, Seoul 08826, Republic of Korea}
\email{yagamelaito@snu.ac.kr}

\author[Ha]{Seung-Yeal Ha}
\address[Seung-Yeal Ha]{\newline Department of Mathematical Sciences and Research Institute of Mathematics \newline Seoul National University, Seoul 08826 and \newline
Korea Institute for Advanced Study, Hoegiro 85, Seoul 02455, Republic of Korea}
\email{syha@snu.ac.kr}

\author[Park]{HanSol Park}
\address[Hansol Park]{\newline Department of Mathematical Sciences\newline Seoul National University, Seoul 08826, Republic of Korea}
\email{hansol960612@snu.ac.kr}

\author[Shim]{Woojoo Shim}
\address[Woojoo Shim]{\newline Department of Mathematical Sciences\newline Seoul National University, Seoul 08826, Republic of Korea}
\email{cosmo.shim@gmail.com }

\newtheorem{theorem}{Theorem}[section]
\newtheorem{lemma}{Lemma}[section]

\newtheorem{proposition}{Proposition}[section]
\newtheorem{remark}{Remark}[section]

\newtheorem{definition}{Definition}[section]

\newcommand{\bx}{\mbox{\boldmath $x$}}
\newcommand{\by}{\mbox{\boldmath $y$}}
\newcommand{\bu}{\mbox{\boldmath $u$}}
\newcommand{\bv}{\mbox{\boldmath $v$}}

\newcommand{\bz}{\mbox{\boldmath $z$}}

\newcommand{\bp}{\mbox{\boldmath $p$}}
\newcommand{\bq}{\mbox{\boldmath $q$}}

\begin{document}

\date{\today}


\thanks{\textbf{Acknowledgment.} The work of S.-Y. Ha was supported by National Research Foundation of Korea(NRF-2020R1A2C3A01003881). f
The work of H. Park was supported by Basic Science Research Program through the National Research Foundation of Korea(NRF) funded by the Ministry of Education (2019R1I1A1A01059585).}
\begin{abstract}
We study emergent behaviors of  Cucker-Smale(CS) flocks on the hyperboloid $\mathbb{H}^d$ in any dimensions. In a recent work \cite{H-H-K-K-M}, a first-order aggregation model on the hyperboloid was proposed and its emergent dynamics was analyzed in terms of initial configuration and system parameters. In this paper, we are interested in the second-order modeling of Cucker-Smale flocks on the hyperboloid. For this, we derive our second-order model from the abstract CS model on complete and smooth Riemannian manifolds by explicitly calculating the geodesic and parallel transport. Velocity alignment has been shown by combining general {velocity alignment estimates} for the abstract CS model on manifolds and verifications of a priori estimate of second derivative of energy functional. For the two-dimensional case $\mathbb{H}^2$, similar to the recent result in \cite{A-H-S}, asymptotic flocking admits only two types of asymptotic scenarios, either convergence to a rest state or a state lying on the same plane (coplanar state). We also provide several numerical simulations to illustrate an aforementioned dichotomy on the asymptotic dynamics of the hyperboloid CS model on $\mathbb{H}^2$.
\end{abstract}

\maketitle \centerline{\date}

%\tableofcontents

\section{Introduction} \label{sec:1}
\setcounter{equation}{0}
Emergent behaviors of many-particle systems are often observed in nature, e.g., flocking of birds \cite{B-C-C, C-S1, C-S2}, aggregation of bacteria \cite{A-B-F, T-B}, swarming of fish \cite{D-M1, D-M2, D-M3, T-T} and synchronization of fireflies and pacemakers \cite{A-R, B-B, Ku2, Pe, Wi2} etc. For survey articles and books on collective dynamics, we refer to \cite{A-B, A-B-F, C-H-L, D-B1, MT2014, P-R,St, Wi1}. In this paper, we continue the study begun in \cite{H-H-K-K-M} on the collective modeling of many-particle systems on the hyperboloid which is one of mathematical model for the hyperbolic space, and we plan to propose a second-order model for CS flocks on the hyperboloid. In \cite{H-H-K-K-M}, authors introduced a first-order Lohe sphere type model on the hyperboloid by replacing the Euclidean inner product with Minkowski one. Moreover, they derived several frameworks leading to the complete aggregations.  

The CS model \cite{C-S1, C-S2} is a Newton-like particle system describing flocking behaviors of many-particle systems, and it has been extensively investigated from various points of view \cite{C-F-R-T, C-D, C-D1, C-S1, C-S2, DS2019, D-Q, D-F-T,H-K-P-Z, H-K-Rug, H-K-Rug2, H-Liu, H-R, H-T, Ji, L-X, M-T, PV17, P-S}. Although extensive studies for the CS model have been done in the last decade, all aforementioned works were mostly restricted to the CS model on Euclidean space. {Therefore, one might ask how does a particle system interacts with the geometric structure of its ambient space, which is an interesting question from the pure and applied mathematics viewpoint. For definiteness, we are interested in the following questions:} 

%\begin{center}
%\textit{how does a particle system interacts with the geometric structure of its ambient space?} 
%\end{center}

\begin{itemize}
\item
Can we generalize the Cucker-Smale model on the Euclidean space to second-order particle models on manifolds? If then, how to formulate the velocity alignment?

\vspace{0.1cm}

\item
For the proposed manifold counterpart of the CS model, can we derive emergent dynamics?
\end{itemize}
Recently, the above questions have been addressed in \cite{H-K-S}, and it has been further extended to the thermodynamic Cucker-Smale(TCS) model by the authors on complete smooth Riemannian manifolds in \cite{A-H-S}.  To fix the idea, we consider a complete and connected Riemannian manifold $M$ with metric tensor $g$, and assume that $(M,g)$ has no boundary. For each $i=1,\cdots,N,$ let $x_i:[0,\infty)\to M$ be a smooth curve representing the position of the $i$-th CS particle on $M$ and $v_i$ be the corresponding velocity of the $i$-th particle. Then, the dynamics of CS particles on $(M,g)$ is governed by the following second-order ODE model \cite{H-K-S}:
\begin{equation}\label{model}
\begin{cases}
\displaystyle {\dot x}_i=v_i,\quad t>0,\quad  i=1,\cdots,N,\\
\displaystyle\nabla_{v_i} v_i=\frac{\kappa}{N}\sum_{j=1}^{N}\psi(x_i, x_j)\left(P_{ij} v_j- v_i\right),\\
\displaystyle(x_i(0), v_i(0))=(x_i^{in}, v_i^{in})\in TM,
\end{cases}
\end{equation}
\noindent where $\nabla$ is the Levi-Civita connection compatible with $(M,g)$, and $P_{ij}$ is the parallel transport along the length minimizing geodesic from $x_j$ to $x_i$, and 
$\kappa$ is a nonnegative coupling strength which is assumed to be strictly positive. Throughout the paper, we call \eqref{model} as the ``{\it abstract manifold CS model}". For a global well-posedness of \eqref{model}, we assume that $\psi$ satisfies positivity, boundedness, symmetry and smoothness conditions:
\begin{align}
\begin{aligned} \label{comm}
& 0 \leq \psi(x, y) = \psi(y, x) \leq \psi_M < \infty, \quad \forall~x,~y \in M, \\
& \psi:M\times M\to\mathbb{R}~\mbox{is a smooth function, and moreover the mapping } \\
&  \{(x_k, v_k)\}_{k=1}^N\mapsto \psi(x_i, x_j)\left(P_{ij} v_j -v_i \right)~\mbox{is smooth for all $i,j=1,\cdots,N$.}
\end{aligned}
\end{align}
\vspace{0.2cm}

\noindent In particular, the smoothness condition on $\psi$ implies that $\psi(x,y)$ has to be zero for each pair $(x,y)\in M\times M$ having more than one length minimizing geodesics, i.e., 
\begin{equation}\label{cut}
\psi(x,y)=0,\quad \forall~y \in \mbox{Cut}(x),
\end{equation}
where Cut($x$) denotes the cut locus of the point $x$, which corresponds to the closure of the set of points $y\in M$ that has more than one length minimizing geodesics joining $x$ and $y$.\\ 

%The model \eqref{model} was first introduced in \cite{H-K-S} as a generalization of Cucker-Smale model to generic complete Riemannian manifolds. 
In \cite{H-K-S}, the smoothness condition in \eqref{comm} was not stated clearly and the communication weight $\psi$ was assumed to be a function of geodesic distances between two points, i.e., 
\begin{equation}\label{previous}
\psi(x,y)=\tilde{\psi}(d(x,y)),
\end{equation}
where $d(x,y)$ denotes a geodesic distance between $x$ and $y$  on $M$. In this case, according to the condition \eqref{cut}, the function ${\psi}(x,y)$ has to be zero, whenever $x$ and $y$ satisfy
\[d(x,y)\in \left\{d(p,q):~p\in \mbox{Cut}(q) \right\}. \]
Therefore, authors in \cite{H-K-S} only considered the case in which $(M,g)$ allows a global injectivity radius $R>0$ and the distance between any two points is less than $R$ along the flow. Indeed, the assumption \eqref{previous} is reasonable when the ambient manifold $(M,g)$ has enough symmetry, e.g. homogeneous spaces such as $\mathbb{R}^d, \mathbb{S}^d$ or $\mathbb{H}^d$. In order to overcome this limitation and study the CS model on a generic complete manifold $(M,g)$, the condition \eqref{previous} for communication weight $\psi$ had been discarded in \cite{A-H-S} and replaced by the smoothness condition \eqref{comm} for a generic $\psi$. In what follows, we set a phase point $z_i$:
\[ z_i := (x_i, v_i), \quad i = 1, \cdots, N. \]
Before we move on further, we recall a concept of asymptotic flocking and velocity alignment as follows.
\begin{definition}\label{D1.1}
\emph{\cite{H-K-S}}
Let $Z=(z_1,\cdots, z_N)$ be a global-in-time smooth solution to \eqref{model}. 
\vspace{0.1cm}
\begin{enumerate}
\item 
The configuration $Z$ exhibits (asymptotic) velocity alignment if 
\begin{equation*} \label{alignment}
\lim_{t \to \infty} \|P_{ij} v_j(t)- v_i(t)\|_{x_i(t)}=0,\quad \forall~ i,j=1,\cdots,N. 
\end{equation*}
\item
The configuration $Z$ exhibits (asymptotic) flocking if
\begin{align}
\begin{aligned} \label{flocking}
\sup_{0 \leq t < \infty} \max_{i,j} d(x_i(t),x_j(t))<\infty,\quad \lim_{t \to \infty} \max_{i,j}  \|P_{ij}v_j(t)-v_i(t)\|_{x_i(t)}=0,
\end{aligned}
\end{align}
\end{enumerate}
{where the norm $\|v\|_x$ in above definition is the canonical norm defined in $T_xM$ associated with inner product $g_x(\cdot,\cdot)$.}

%will be explained in next section.
\end{definition}

\bigskip

Now, we return to our manifold setting with $M = {\mathbb H}^d$. Consider a smooth solution of system \eqref{model} on $\mathbb{H}^d$ and denote the position configuration as $X:=\{ x_1,\cdots, x_N \}$. For each curve $x_i$, we then write their $(d+1)$-dimensional coordinate expression as $\bx_i=\iota(x_i)$. Similarly, we also consider the coordinate expression of tangent vector $\dot{x}_i$ in $\mathbb{R}^{d+1}$ by using the canonical orthonormal basis of $T\mathbb{R}^{d+1}$:
\begin{equation}\label{tangent}
\bx_i=\iota(x_i),\quad \bv_i:=(v_i^0,\cdots,v_i^d),\quad\mbox{where}\quad d\iota_{x_i}(\dot{x}_i)=\sum_{j=0}^{d}v_i^j\frac{\partial}{\partial x^j}\Big|_{\iota(x_i)}. 
\end{equation} 
%Note that we used bold-faced letters to denote points and tangent vectors when they can be viewed as corresponding vectors in $\bbr^{d+1}$. For the sake of simplicity, we here write 
%\[\psi(\bx_i,\bx_j):=\psi(x_i,x_j)\quad  \mbox{and}\quad  \|\bv_i\|^2:=\langle\bv_i,\bv_i\rangle_M , \] 
%for points $\bx_i,\bx_j\in\iota(\mathbb{H}^d)$ and a tangent vector $\bv_i$.

 In this setting, the abstract manifold CS model \eqref{model} reduces to the following explicit form:
\begin{equation}\label{model-1}
\begin{cases}
\displaystyle \dot{\bx}_i=\bv_i,\quad t>0,\quad  i=1,\cdots,N,\\
\displaystyle \dot{\bv}_i=||\bv_i||^2{\bx}_i+\frac{\kappa}{N}\sum_{j=1}^{N}\psi(\bx_{i}, \bx_{j} )\left({\bv_j}-{\bv_i}+\frac{\langle \bx_{i}, \bv_{j} \rangle_M }{1-\langle \bx_{i}, \bx_{j} \rangle_M }(\bx_{i}+\bx_{j})\right),\\
\displaystyle(\bx_i(0),\bv_i(0))=(\bx_i^{in},\bv_i^{in})\in \mathbb{R}^{d+1}\times \mathbb{R}^{d+1},\quad i=1,\cdots,N,
\end{cases}
\end{equation}
subject to the initial constraints: 
\begin{equation}\label{model-ini}
\langle \bx_i^{in},\bx_i^{in}\rangle_M =-1,\quad \langle \bx_i^{in},\bv_i^{in}\rangle_M =0,\quad i=1,\cdots,N. 
\end{equation}
{We refer Section \ref{sec:2} and \ref{sec:3} for the definition of $\|\cdot\|,~ \langle\cdot,\cdot\rangle_M$ and a detailed derivation from \eqref{model} to \eqref{model-1}-\eqref{model-ini}.} Throughout the paper, we call system \eqref{model-1}-\eqref{model-ini} as the hyperbolic Cucker-Smale(HCS) model.  \newline

Next, we briefly discuss {our} three main results. First, we show that when the HCS model \eqref{model-1} is restricted on a geodesic, it reduces to the hyperbolic Kuramoto(HK) model introduced in \cite{R-L-W} (see Section \ref{sec:3}):
\[
\dot{\theta}_i=\nu_i+\frac{\kappa}{N}\sum_{j=1}^{N}\sinh(\theta_j-\theta_i),\quad i=1,\cdots,N.
\]
Second, we show that if $\psi$ in $\eqref{model-1}_2$ has a positive lower bound, system \eqref{model-1} exhibits asymptotic velocity alignment in the sense of Definition \ref{D1.1} ({see} Theorem \ref{T4.1}): 
\[\lim_{t \to \infty} \|P_{ij} v_j- v_i\|_{x_i}=0,\quad \forall~ i,j=1,\cdots,N. \]
Third, we derive a dichotomy on the asymptotic patterns of the HCS model on ${\mathbb H}^2$ under the a priori flocking assumption \eqref{flocking}. For this, we introduce an energy functional ${\mathcal E}$:
\begin{equation} \label{energy}
 \mathcal{E}[t]:= \frac{1}{2} \sum_{i=1}^{N}\|v_i\|_{x_i}^2 = \frac{1}{2} \sum_{i=1}^{N}g({\dot x}_i, {\dot x}_i), \quad t \geq 0. 
\end{equation} 
Then, it is known \cite{H-K-S} that $\mathcal{E}[t]$ is monotonically decreasing along the abstract {manifold CS} model \eqref{model}{, which includes the HCS flow \eqref{model-1} as a particular case $M=\mathbb{H}^d$.} Hence, it converges to a nonnegative value asymptotically (see Proposition \ref{P4.1}). Then, depending on the limit, we have the following two cases:
\[ \mbox{either}~ \lim_{t\to\infty}\mathcal{E}[t]=0 \quad \mbox{or} \quad \lim_{t\to\infty}\mathcal{E}[t]>0. \]
For the latter case, via several technical lemmas (Lemma \ref{L4.2} - Lemma \ref{L4.4}) on the hyperbolic geodesic triangles, we show that position configuration becomes coplanar asymptotically (see Theorem \ref{T4.2}).

\bigskip

The rest of this paper is organized as follows. In Section \ref{sec:2}, we study the CS model on {generic Riemannian manifolds} and review basic materials on the hyperbolic space.  In Section \ref{sec:3}, we present the hyperbolic Kuramoto model {and derive it from the HCS model on a geodesic.} {In Section} \ref{sec:4}, we study emergent dynamics of the HCS model, e.g. velocity alignment, asymptotic flocking and dichotomy in asymptotic dynamics and {provide several numerical simulations to compare} them with {our} analytic results. Finally Section \ref{sec:5} is devoted to a brief summary of our main results and some discussion on the remaining issues to be explored in a future work.

\section{Preliminaries} \label{sec:2}
\setcounter{equation}{0}
In this section, we review theoretical minimum for the hyperboloid model as a hyperbolic space and a priori velocity alignment estimate for the abstract CS model on complete and smooth Riemannian manifolds, and then provide explicit CS model on the hyperboloid. 
\subsection{The hyperbolic space $\mathbb{H}^d$}\label{sec:2.1}
%In this subsection, we briefly discuss some theoretical backgrounds and several basic properties on $d$-dimensional hyperbolic space $\mathbb{H}^d$.
%\subsubsection{The hyperbolic space}
%The hyperbolic $d$-space, denoted by $\mathbb{H}^d$, is a $d$-dimensional Riemannian manifold which is complete, connected, simply-connected with constant sectional curvature $-1$. 
%Diffeomorphic to $\mathbb{R}^d$ but has a different metric.\\
%Euclidean geometry...\\
%Elliptic geometry...\\
%Hyperbolic geometry...
%\begin{theorem}\emph{(Killing \cite{Ki}, Hopf \cite{H})}
%	Every $d$-dimensional complete connected Riemannian manifolds of constant sectional curvature are isometric to a quotient of a sphere $\mathbb{S}^d$, Euclidean space $\mathbb{R}^d$, or hyperbolic space $\mathbb{H}^d$ by free and properly discontinuous group actions.
%\end{theorem}
%Therefore, the hyperbolic space $\mathbb{H}^d$ can be also regarded as the universal cover of arbitrary Riemannian manifold with constant sectional curvature $-1$.\\
%The Poincaré ball model....\\
%The Poincaré half space model....\\
%The Klein model....\\
%The Hyperboloid model....\\
%\textcolor{blue}{This part will be completed by Woojoo Shim}
\vspace{0.1cm}

The hyperbolic $d$-space, denoted by $\mathbb{H}^d$, is a {unique} $d$-dimensional Riemannian manifold which is complete, connected, simply-connected with constant sectional curvature $-1$ {up to isometry}. In the sequel, we characterize the hyperbolic space $\mathbb{H}^d$ for the hyperboloid model with some geometric properties. In the hyperboloid model,  $\mathbb{H}^d$ is given by
\vspace{0.1cm}
\begin{equation}\label{set}
\mathbb{H}^d=\left\{(x^0, x^1, \cdots, x^d)\in\mathbb{R}^{d+1}: -(x^0)^2+(x^1)^2+\cdots+(x^d)^2=-1, \quad x^0>0\right\},
\end{equation}
{where} the smooth structure of $\mathbb{H}^d$ can be described by a single chart $(\mathbb{H}^d, \phi)$ for the homeomorphism $\phi$:
\begin{align*}
\phi&:\mathbb{H}^d\rightarrow\mathbb{R}^d,\quad \phi(x^0, x^1, \cdots, x^d)=(x^1, x^2, \cdots, x^d) \quad \mbox{and} \\
\phi^{-1}&:\mathbb{R}^d\rightarrow\mathbb{H}^d,\quad \phi^{-1}(x^1, \cdots, x^d)=\left(\sqrt{1+(x^1)^2+\cdots+(x^d)^2},x^1, x^2, \cdots, x^d\right).
\end{align*}
Then, $\mathbb{H}^d$ is connected and simply connected since it is homeomorphic to $\mathbb{R}^d$.  In the following proposition, we state {a well-known fact} for the geodesic of $\mathbb{H}^d$, which guarantees the completeness of $\mathbb{H}^d$. To make the notion clear, we now denote $\iota$ as the natural injection from $\mathbb{H}^d$ to $\mathbb{R}^{d+1}$.
\begin{proposition} \label{P2.1}
	Let $\gamma:\mathbb{R}\to\mathbb{H}^d$ be a geodesic curve on $\mathbb{H}^d$ parametrized by arc length. Then, $\Gamma:=\iota\circ\gamma$ can be written as 
	\begin{align*}
	\begin{aligned}
	\Gamma(s)= \cosh(s)\cdot\Gamma(0) + \sinh(s)\cdot\dot{\Gamma}(0),\quad \forall~s\in\mathbb{R}.  
	\end{aligned}
	\end{align*}
\end{proposition}
\begin{proof}
%Since the proof is rather lengthy, we leave it in Appendix \ref{App-A}.
Although this is well-known fact for $\mathbb{H}^d$, we present its proof in Appendix \ref{App-A} for completeness.
\end{proof}
\begin{remark}
	In the unit sphere $\mathbb{S}^d$, any geodesic $\gamma:\mathbb{R}\to \mathbb{S}^d\subset \mathbb{R}^{d+1}$ satisfies
		\begin{align*}
	\begin{aligned}
	\Gamma(s)= \cos(s)\cdot\Gamma(0) + \sin(s)\cdot\dot{\Gamma}(0),\quad \forall~s\in\mathbb{R}.  
	\end{aligned}
	\end{align*}

\end{remark}
%\noindent Now, we find a sectional curvature of $\mathbb{H}^d$ for arbitrary two-dimensional subspace.
% Moreover, one can also easily check that this $\gamma$ indeed gives a unique (and therefore length-minimizing) geodesic on $\mathbb{H}^d$ joining any pair of points on $\gamma$.

Next, we discuss some properties of a geodesic curve on ${\mathbb{H}^d}$ and parallel transports along length minimizing geodesics on ${\mathbb{H}^d}$ by using Minkowski bilinear form. We first begin with definition of Minkowski bilinear form, denoted by $\langle\cdot,\cdot\rangle_M $. 
\vspace{0.2cm}
\begin{definition}
	The Minkowski bilinear form $\langle\cdot,\cdot\rangle_M :\mathbb{R}^{d+1}\times \mathbb{R}^{d+1}\to\mathbb{R}$ is defined as
	\[ \langle \bm{x},\bm{y} \rangle_M  := -x^0y^0+\sum_{k=1}^{d}x^{k}y^{k} \in \mathbb{R}, \]
where
\[ \bm{x}=~(x^0, x^1, \cdots, x^d)\in \mathbb{R}^{d+1}, \quad \bm{y}=(y^0, y^1, \cdots, y^d)\in \mathbb{R}^{d+1}. \]
	
\end{definition}
This bilinear form can be extensively used to simplify several notions in the hyperboloid model. Among them, we here introduce  some well-known properties of Minkowski bilinear form in the following three propositions to be crucially related to analysis in later sections. The first one gives some defining relations of points and tangent vectors on the hyperboloid model in terms of Minkowski bilinear form. 
\vspace{0.2cm}
\begin{proposition}\label{P2.2}
	Let $\bx=(x^0,\cdots,x^d),~\bu=(u^0,\cdots,u^d)$ and $\bv=(v^0,\cdots,v^d)$ be arbitrary vectors in $\mathbb{R}^{d+1}$. Then, the following assertions hold:
	\vspace{0.2cm}
	\begin{enumerate}
		\item $\bm{x}$ lies on $\iota({\mathbb{H}^d})$ if and only if 
		\[\langle \bm{x},\bm{x} \rangle_M =-1\quad\mbox{and}\quad x^0>0.\]
		\item $\sum_{i=0}^d v^i\frac{\partial}{\partial{x^i}}\big|_{\bx}$ is contained in $d\iota_x \left(T_{x}\mathbb{H}^d\right)$ if and only if 
		\[\bx=\iota(x)\quad \mbox{and}\quad \langle \bx,\bm{v} \rangle_M =0,\]
		 where $\left\{\frac{\partial}{\partial{x^0}}\big|_{\bx},\cdots, \frac{\partial}{\partial{x^d}}\big|_{\bx}\right\}$ is the canonical orthomormal basis of $T_{\bx}\mathbb{R}^{d+1}$ at point $\bx$.
		\item Let $x$ be a point in $\mathbb{H}^d$ and $u,v$ be two tangent vectors on $\mathbb{H}^d$ at $x$. If $\bx,\bu$ and $\bv$ satisfy \[\bm{x}=\iota(x)\quad \mbox{and} \quad\sum_{i=0}^d u^i\frac{\partial}{\partial{x^i}}\Big|_{\bx}=d\iota_x(u),~~\,\, \sum_{i=0}^d v^i\frac{\partial}{\partial{x^i}}\Big|_{\bx}=d\iota_x(v),\]  we have an equivalence between $\langle\cdot,\cdot\rangle_M $ and $g(\cdot,\cdot)$:
		\begin{equation}\label{identity}
		\langle \bm{u},\bm{v} \rangle_M =g\left({u},{v}\right). 
		\end{equation} 
	\end{enumerate}
\end{proposition}
\vspace{0.1cm}
\begin{proof}
	The first assertion immediately follows from definition \eqref{set}, and we deduce the second one by taking derivative of $(1)$. Therefore, we only present a proof of $(3)$. For $(3)$, it suffices to verify \eqref{identity} for the basis 
	\[\left\{\partial_1,\cdots,\partial_d\right\}:=\left\{d\phi^{-1}\left(\frac{\partial}{\partial x^1}\Big|_{\phi(x)}\right),\cdots,d\phi^{-1}\left(\frac{\partial}{\partial x^d}\Big|_{\phi(x)}\right) \right\}, \]
	since $g(\cdot,\cdot)$ and $\langle\cdot,\cdot\rangle_M $ are both bilinear. Then, we use
	\[ 
	d\iota_x(\partial_i) =d(\iota_x\circ\phi^{-1})\left(\frac{\partial}{\partial x^i}\Big|_{\phi(x)}\right)\\
	=\frac{x^i}{\sqrt{1+(x^1)^2+\cdots+(x^d)^2}}\frac{\partial}{\partial x^0}\Big|_{\bx}+\frac{\partial}{\partial x^i}\Big|_{\bx} \]
	to deduce
	\[
	\begin{aligned}
	g(\partial_i,\partial_j) &=g_{ij}(\phi(x)) =\delta_{ij}-\frac{x^ix^j}{1+\sum_{m=1}^d(x^m)^2}\\
	&=\left\langle \frac{x^i}{\sqrt{1+\sum_{m=1}^d(x^m)^2}}e_0+e_i, \frac{x^j}{\sqrt{1+\sum_{m=1}^d(x^m)^2}}e_0+e_j\right\rangle_M ,
	\end{aligned} \]
	where $e_k=(0,\cdots,0,1,0,\cdots,0)$ is a unit vector which has a unique nonzero value at $(k+1)$-th component.
\end{proof}
\vspace{0.2cm}
Meanwhile, the bilinear form $\langle\cdot,\cdot\rangle_M $ can be also used to represent the geodesic distance between points and show the uniqueness of geodesic joining arbitray pair of points. 
\vspace{0.2cm}
\begin{proposition} \label{P2.3}
	For every $p\neq q\in\mathbb{H}^d$, there exists a unique geodesic joining $p$ and $q$ up to reparametrization. Moreover, their geodesic distance $d(p,q)$ satisfies
    \[
    \cosh d({p},{q})= -\langle \bp,\bq \rangle_M ,\quad \bp:=\iota({p}),~\bq=\iota({q}).
    \]    
\end{proposition}
\begin{proof}
	Suppose $\gamma_1$ and $\gamma_2$ are two geodesics from $p$ to $q$ parametrized by arclength. Then, we use Proposition \ref{P2.1} to write
	\vspace{0.2cm}
	\begin{equation}\label{twogeods}
	\begin{aligned}
	\Gamma_i(s_i)&= \cosh(s_i)\cdot\Gamma_i(0) + \sinh(s_i)\cdot\dot{\Gamma}_i(0),\\
	\gamma_i(0)&=p,\quad \gamma_i(s_i)=q,\quad \Gamma_i:=\iota\circ\gamma_i, \quad i=1,2,
	\end{aligned}
	\end{equation} 
	\vspace{0.2cm}
	
	\noindent where $s_1,s_2>0$ are the lengths of geodesics $\gamma_1$ and $\gamma_2$ from $p$ to $q$, respectively. Now, we substitute the representations \eqref{twogeods} of $p$ and $q$ to $\langle\bp,\bq\rangle_M $ to obtain
	\vspace{0.2cm}
	\[\begin{aligned}
	\langle\bp,\bq\rangle_M &=\langle\Gamma_i(0),\Gamma_i(s_i) \rangle_M \\
	&= \cosh(s_i)\langle\Gamma_i(0),\Gamma_i(0) \rangle_M +\sinh(s_i)\langle\Gamma_i(0),\dot{\Gamma}_i(0) \rangle_M \\
	&=-\cosh(s_i),\quad i=1,2,
	\end{aligned} \]
\noindent where we used Proposition \ref{P2.2} in the last equality.  \newline
	
\noindent	Therefore, we deduce $s_1=s_2$ and 
	\vspace{0.2cm}
	\[\dot{\Gamma}_1(0)=\frac{\bq-\cosh(s_1)\bp}{\sinh(s_1)}=\frac{\bq-\cosh(s_2)\bp}{\sinh(s_2)}=\dot{\Gamma}_2(0), \]
	\vspace{0.2cm}
	\noindent which is equivalent to $\gamma_1=\gamma_2$.
\end{proof}
\begin{remark}\label{R2.2}
By the result of Proposition \ref{P2.3}, the hyperbolic space $\mathbb{H}^d$ admits a unique geodesic for each pair of points. Therefore, the mapping 
	\[\{(x_k, v_k)\}_{k=1}^N\mapsto \psi(x_i,x_j)\left(P_{ij} v_j- v_i \right) \]
	is smooth if and only if $\psi$ is smooth, and therefore any nonnegative, bounded, symmetric and smooth $\psi(\cdot,\cdot)$ can be used in the hyperbolic Cucker-Smale model \eqref{model-1}-\eqref{model-ini}. In particular, $\psi$ may have a positive lower bound or it is even possible to take $\psi$ as a constant without any violation of well-posedness.
\end{remark}
Finally, Minkowski bilinear form $\langle\cdot,\cdot\rangle_M $ can be used to write down the covariant derivative $\nabla_{\dot{x}}\dot{x}$ and parallel transport along geodesics explicitly, which is necessary to analyze \eqref{model} more comprehensively.

\begin{proposition} \label{P2.4}
	Let ${\mathbb{H}^d}$ and  $\nabla$ be the hyperbolic space represented by the hyperboloid model and the Levi-Civita connection of $({\mathbb{H}^d},g)$, respectively. Then, the following assertions hold.
	\begin{enumerate}
		\item For any smooth curve $x:\mathbb{R}\to \mathbb{H}^d$, the covariant derivative of tangent vector field $\dot x$ along $x$ satisfies
		\vspace{0.2cm}
		\begin{equation}\label{covariant}
		\begin{aligned}
		\nabla_{\dot{\bx}}\dot{\bx}=\ddot{\bx}-\langle
		\dot{\bx},\dot{\bx}\rangle_M \bx=\ddot{\bx}+\langle
		\ddot{\bx},\bx\rangle_M \bx,
		\end{aligned}
		\end{equation}
		\vspace{0.1cm}
		 
		\noindent where $\bx=\iota\circ x=(x^0,\cdots,x^d)$ is a curve in $\mathbb{R}^{d+1}$ and 
		\[\nabla_{\dot{\bx}}\dot{\bx}:=(a^0,\cdots,a^d),\quad d\iota_x(\nabla_{\dot{x}}\dot{x})=\sum_{i=0}^{d}a^i\frac{\partial}{\partial x^i}\Big|_{\bx}. \]
		\item Let  $p$ and $q$ be two distinct points in $\mathbb{H}^d$, and $\gamma$ be the geodesic from $p$ to $q$. If $\bu$ is the parallel transport of $\bv$ from $\bp = \iota(p)$ to $\bq = \iota(q)$ along $\gamma$, we have
		\vspace{0.2cm}
		\begin{equation}\label{parallel}
	\bu=\bv+\frac{\langle \bv, \bq \rangle_M }{1-\langle \bp,\bq\rangle_M }(\bp+\bq),
		\end{equation}
		\vspace{0.1cm}
		
		\noindent where $\bv=(v^0,\cdots,v^d)$ and $\bu=(u^0,\cdots,u^d)$ are given by the relations:
		\[\begin{aligned}
		d\iota_p(v)=\sum_{i=0}^{d}v^i\frac{\partial}{\partial x^i}\Big|_{\bp},\quad  d\iota_q(u)=\sum_{i=0}^{d}u^i\frac{\partial}{\partial x^i}\Big|_{\bq}.
		\end{aligned} \]
	\end{enumerate}
	\begin{proof}
	Since the proof is rather lengthy, we leave it in Appendix \ref{App-B}.
	\end{proof}
\end{proposition}

%\section{From the HCS model to the HK model} \label{sec:3}
%\setcounter{equation}{0}
%In this section, we \textcolor{red}{recall} the hyperbolic Kuramoto model \cite{R-L-W} and \textcolor{red}{show} how it can be derivable from the HCS model on a geodesic.

\section{{From HCS to HK model}} \label{sec:3}
In this section, we provide an explicit representation for the CS model on the hyperboloid.
%, thanks to the \textcolor{red}{$(d+1)$-dimensional coordinate expressions} for the Levi-Civita connection and parallel transport given in Proposition \ref{P2.4}. 
{Then, we {recall} the hyperbolic Kuramoto model \cite{R-L-W} and  provide a rigorous derivation of HK model from the HCS model on a geodesic. }  \newline

\noindent  Recall that the $(d+1)$-dimensional coordinate expression of $(x_i,v_i)\in TM$ is given by
%Consider a smooth solution of system \eqref{model} on $\mathbb{H}^d$, and denote the position configuration as $X:=\{ x_1,\cdots, x_N \}$. For each curve $x_i$, we then write their $(d+1)$-dimensional coordinate expression as $\bx_i=\iota(x_i)$. Similarly, we also consider the coordinate expression of tangent vector $\dot{x}_i$ in $\mathbb{R}^{d+1}$ by using the canonical orthonormal basis of $T\mathbb{R}^{d+1}$:
\begin{equation*}%\label{tangent}
\bx_i=\iota(x_i),\quad \bv_i:=(v_i^0,\cdots,v_i^d),\quad\mbox{where}\quad d\iota_{x_i}(\dot{x}_i)=\sum_{j=0}^{d}v_i^j\frac{\partial}{\partial x^j}\Big|_{\iota(x_i)}. 
\end{equation*} 
By abuse of notation, we used $\psi(\cdot,\cdot)$ and $\|\cdot\|$ as:
\[\psi(\bx_i,\bx_j):=\psi(x_i,x_j) \quad \mbox{and}\quad  \|\bv_i\|^2:=\langle\bv_i,\bv_i\rangle_M , \] 
for points $\bx_i,\bx_j\in\iota(\mathbb{H}^d)$ and a tangent vector $\bv_i$. \newline

\noindent Now we use Proposition \ref{P2.2} and \ref{P2.4} to reduce system \eqref{model} {to \eqref{model-1}-\eqref{model-ini}}: %in the following system of ODEs:
\begin{equation}\label{hcs}
\begin{cases}
\displaystyle \dot{\bx}_i=\bv_i,\quad t>0,\quad  i=1,\cdots,N,\\
\displaystyle \dot{\bv}_i=||\bv_i||^2{\bx}_i+\frac{\kappa}{N}\sum_{j=1}^{N}\psi(\bx_{i}, \bx_{j} )\left({\bv_j}-{\bv_i}+\frac{\langle \bx_{i}, \bv_{j} \rangle_M }{1-\langle \bx_{i}, \bx_{j} \rangle_M }(\bx_{i}+\bx_{j})\right),\\
\displaystyle(\bx_i(0),\bv_i(0))=(\bx_i^{in},\bv_i^{in})\in \mathbb{R}^{d+1}\times \mathbb{R}^{d+1},\quad i=1,\cdots,N,
\end{cases}
\end{equation}
together with the initial constraint 
\begin{equation}\label{hcsinitial}
\langle \bx_i^{in},\bx_i^{in}\rangle_M =-1,\quad \langle \bx_i^{in},\bv_i^{in}\rangle_M =0,\quad i=1,\cdots,N. 
\end{equation}
Since \eqref{hcs}--\eqref{hcsinitial} is an alternative representation of intrinsic ordinary differential equation \eqref{model} on $\mathbb{H}^d$, it is clear that $\bx_1,\cdots\bx_N$ lie on $\iota(\mathbb{H}^d)$ along the flow.  From now on, we call the model as ``{\it the hyperbolic Cucker-Smale (HCS) model}".
\\\\
{On the other hand,} HK model was first formulated in \cite{R-L-W} as a hyperbolic analogue of classical Kuramoto model \cite{Ku2}, which has been extensively studied as a prototype mathematical model of sychronization for weakly coupled oscillators. 
More precisely, it has been derived by selecting the noncompact Lorentz Group $SO(1,1)$ as the symmetry group for the system of matrix equation {in \cite{H-K-R}:}
\begin{equation}\label{Lie}
\dot{X}_iX_i^{-1}=\Omega_i+\frac{\kappa}{2N}\sum_{j=1}^{N}\left(X_jX_i^{-1}-X_iX_j^{-1} \right),
\end{equation}
where the group element $X_i$ is a position of $i$-th node and $\Omega_i$ is an element of corresponding Lie algebra. Then, we parametrize $X_i\in SO(1,1)$ and $\Omega_i$ to 
\begin{equation} \label{New-1}
X_i=e^{\alpha_i J}=\begin{pmatrix}
\cosh \alpha_i~ \sinh \alpha_i\\\sinh \alpha_i~\cosh \alpha_i
\end{pmatrix},\quad \Omega_i:=\omega_i J, \quad J:=\begin{pmatrix}
0\quad 1\\1\quad 0
\end{pmatrix}, 
\end{equation}
{for hyperbolic angles $\alpha_i$ and real parameters $\omega_i$, and we substitute ansatz \eqref{New-1} into \eqref{Lie} to get the following system:}
\begin{equation}\label{hk}
\dot{\alpha}_i=\omega_i+\frac{\kappa}{N}\sum_{j=1}^{N}\sinh(\alpha_j-\alpha_i),\quad i=1,\cdots,N,\quad t>0.
\end{equation}
This system is coined as the hyperbolic Kuramoto model in {\cite{R-L-W}.} In addition, the equation \eqref{Lie} also exhibits an asymptotic phase locking:
\[\exists~~\lim_{t \to \infty}X_iX_j^{-1}\, \quad \forall~i,j=1,\cdots,N, \]
for sufficiently large $\kappa$ and for restricted initial configuration. 
\\\\
{Now}, we discuss a reduction from the HCS model to the HK model. For this, we consider the HCS model on a geodesic.  Suppose that there exists a geodesic $\gamma$ of $\mathbb{H}^d$ containing all initial positions $x_1^{in},\cdots,x_N^{in}$ of the CS model \eqref{model}. We further assume that all initial velocity vectors $v_1^{in},\cdots,v_N^{in}$ are tangent to the geodesic $\gamma$ so that the geodesic $\gamma$ is positively invariant under the HCS flow. Then, system \eqref{model} becomes intrinsic to the submanifold $\gamma$ of $\mathbb{H}^d$, and we can parametrize the time-evolution of $x_i$ by one-dimensional parameter $\alpha_i$ (see Proposition \ref{P2.1}):
\[\bx_i(s)=\iota(x_i(s))=\bp\cosh(\alpha_i(s))+\bq\sinh(\alpha_i(s)), \]
where $\bp$ and $\bq$ satisfy
\[\langle\bp,\bp\rangle_M =-1,\quad \langle\bp,\bq\rangle_M =0,\quad \langle\bq,\bq\rangle_M =1.\]
Then, the first and second derivative of $\bx_i$ can be written as 
\[
\begin{aligned}
\dot{\bx}_i&=\left(\bp\sinh\alpha_i+\bq\cosh\alpha_i \right)\dot{\alpha}_i,\quad \|\dot{\bx}_i\|=|\dot{\alpha}_i|,\quad \\
\ddot{\bx}_i&=\left(\bp\sinh\alpha_i+\bq\cosh\alpha_i \right)\ddot{\alpha}_i+\left(\bp\cosh\alpha_i+\bq\sinh\alpha_i\right)|\dot{\alpha}_i|^2,
\end{aligned}
\]
and we have 
\[\langle \bx_i,\bx_j\rangle_M =-\cosh(\alpha_i-\alpha_j),\quad \langle\bx_i,\dot{\bx}_j\rangle_M =\sinh(\alpha_i-\alpha_j)\dot{\alpha}_j. \]
Therefore, we reduce the hyperbolic CS model \eqref{hcs} to the following ODE system for the parameters $(\alpha_i)_i$: for $i=1,\cdots,N$,
\begin{align}\label{e1}
\begin{aligned}
&\left(\bp\sinh\alpha_i+\bq\cosh\alpha_i \right)\ddot{\alpha}_i \\ &=\frac{\kappa}{N}\sum_{j=1}^{N}\psi(\bx_i,\bx_j)\left(\dot{\bx}_j-\dot{\bx}_i+\frac{\sinh(\alpha_i-\alpha_j)}{1+\cosh(\alpha_i-\alpha_j)}\dot{\alpha}_j(\bx_i+\bx_j) \right).
\end{aligned}
\end{align}
We then apply a linear operator $\left\langle \left(\bp\sinh\alpha_i+\bq\cosh\alpha_i \right),\cdot\right\rangle_M $ to the both sides of \eqref{e1} to deduce
 \begin{equation}\label{e2}
\begin{aligned}
 \ddot{\alpha}_i&=\frac{\kappa}{N}\sum_{j=1}^{N}\psi(\bx_i,\bx_j)\left(\cosh(\alpha_i-\alpha_j)\dot{\alpha_j}-\dot{\alpha}_i+\frac{\sinh(\alpha_i-\alpha_j)}{1+\cosh(\alpha_i-\alpha_j) }\dot{\alpha}_j\cdot\sinh(\alpha_j-\alpha_i) \right)\\
&=\frac{\kappa}{N}\sum_{j=1}^{N}\psi(\bx_i,\bx_j)\left(\dot{\alpha}_j-\dot{\alpha}_i \right),\quad i=1,\cdots,N.
\end{aligned}
 \end{equation}
Since $\mathbb{H}^d$ is a homogeneous space for the Lorentz group $O^+(1,d)$ as an isometry group, it is reasonable to assume $\psi(\bx_i,\bx_j)$ as a function of geodesic distance $d(x_i,x_j)=\cosh(\alpha_i-\alpha_j)$ from the symmetry of $\mathbb{H}^d$. In particular, if $\psi(\bx_i,\bx_j)=d(x_i,x_j)$, we can integrate the right-hand side of \eqref{e2} and obtain the first-order ODE system 
 \[\begin{aligned}
 \dot{\alpha}_i&=\left[\dot{\alpha}_i(0)-\frac{\kappa}{N}\sum_{j=1}^{N}\sinh(\alpha_j(0)-\alpha_i(0))\right]+\frac{\kappa}{N}\sum_{j=1}^{N}\sinh(\alpha_j-\alpha_i)\\
 &=:\omega_i+\frac{\kappa}{N}\sum_{j=1}^{N}\sinh(\alpha_j-\alpha_i),\quad i=1,\cdots,N,
 \end{aligned}   \]
 {which} exactly coincides with the hyperbolic Kuramoto model \eqref{hk}.

 \section{Emergent behaviors for the HCS model} \label{sec:4}
\setcounter{equation}{0}
In this section, we first recall velocity alignment for the abstract manifold CS model and provide an improved velocity alignment estimate for $\mathbb{H}^d$. 

\subsection{Abstract manifold CS model} \label{sec:4.1} 
In this subsection, we briefly recall the velocity alignment estimate in \cite{A-H-S, H-K-S} for the abstract CS model \eqref{model} on manifold $(M,g)$:
\begin{equation}\label{MCS}
\begin{cases}
\displaystyle \dot{x}_i= v_i,\quad t>0,\quad  i=1,\cdots,N,\\
\displaystyle\nabla_{v_i} v_i=\frac{\kappa}{N}\sum_{j=1}^{N}\psi(x_i,x_j)\left({P_{ij}} v_j- v_i \right).
\end{cases}
\end{equation}
To see the dissipative nature of \eqref{MCS} and velocity alignment in a priori setting, we recall the energy functional defined in \eqref{energy}:
\[ \mathcal{E}[t]:= \frac{1}{2} \sum_{i=1}^{N}\|v_i\|_{x_i}^2= \frac{1}{2} \sum_{i=1}^{N}g({\dot x}_i, {\dot x}_i), \quad t \geq 0. \] 
\begin{proposition}\emph{\cite{A-H-S, H-K-S}}\label{P4.1}
	Let $Z:=(z_1,\cdots,z_N)$ be a global-in-time smooth solution to \eqref{MCS}. Then, the following two assertions hold.
\begin{enumerate}
\item
The total kinetic energy $ \mathcal{E}[\cdot]$ is monotonically decreasing in time:
	\[\frac{d}{dt} \mathcal{E}[t]=-\frac{\kappa}{2N}\sum_{i,j}\psi(x_i,x_j)\|P_{ij} v_j- v_i\|_{x_i}^2\leq 0, \quad t > 0. \] 
\item
Suppose that the following a priori conditions hold:
	\begin{equation}\label{CSap}
	\inf_{0\leq t<\infty}\min_{i,j}\psi(x_i(t), x_j(t))=:\psi_m>0,\quad \sup_{0 \leq t < \infty} \left|\frac{d}{dt}\sum_{i, j}\|P_{ij} v_j- v_i\|_{x_i}^2 \right|<\infty.
	\end{equation}
	Then, the configuration $Z=Z(t)$ exhibits asymptotic velocity alignment: 
	\[\lim_{t \to \infty} \|P_{ij} v_j- v_i\|_{x_i}=0,\quad \forall~ i,j=1,\cdots,N. \]
\end{enumerate}
\end{proposition}
\begin{proof} 
The detailed proofs can be found in \cite{A-H-S, H-K-S}, but for reader's convenience, we briefly sketch the proof below.\\ \newline
Then dissipativity and a priori velocity alignment for \eqref{MCS} can be summarized in the following proposition.
\noindent (i)~We use $\eqref{MCS}_2$ and index exchange $i \longleftrightarrow k$ to derive the dissipation estimate for ${\mathcal E}$:
\begin{align*}
\begin{aligned}
\frac{d {\mathcal E}}{dt}	&= \sum\limits_{i=1}^N g_{x_i}\Big(v_i, \frac{\kappa}{N} \sum\limits_{k=1}^N \psi(x_i,x_k) (P_{ik} v_k- v_i) \Big)_{x_i}\\
					&= \frac{\kappa}{2N} \sum\limits_{i,k=1}^N \psi(x_i,x_k) \left(g_{x_i}\Big( v_i,  P_{ik} v_k-v_i\Big)+g_{x_k}\Big( v_k,  P_{ki} v_i- v_k  \Big)\right)\\
					&= \frac{\kappa}{2N} \sum\limits_{i,k=1}^N \psi(x_i,x_k) \left(g_{x_i}\Big( v_i,  P_{ik} v_k-v_i\Big)+g_{x_i}\Big( P_{ik} v_k, v_i-P_{ik} v_k \Big)\right)\\
					&= -\frac{\kappa}{2N} \sum\limits_{i,k=1}^N \psi(x_i,x_k) \lVert P_{ik}v_k-v_i\rVert_{x_i}^2 \leq0,
\end{aligned}
\end{align*}
where we also used the following properties of parallel transport:
\[ P_{ik} P_{ki}= I \quad \mbox{and} \quad  g_{x_i}\left(P_{ik} v_k, P_{ik} w_k\right) =  g_{x_k}\left(v_k, w_k\right). \]

\vspace{0.2cm}

\noindent (ii)~It follows  from the first assertion that
	\[\int_{0}^{T}\left(\sum_{i,j}\psi(x_i, x_j)\|P_{ij} v_j- v_i\|_{x_i}^2\right) dt=\frac{2N}{\kappa}\left(\mathcal{E}(0)-\mathcal{E}(T) \right)\leq \frac{2N\mathcal{E}(0)}{\kappa}<\infty,\quad \forall~ T>0. \]
	Then, \textit{a priori} assumption $\eqref{CSap}_1$ implies 
	\[\int_{0}^{\infty}\sum_{i, j}\|P_{ij} v_j- v_i\|^2_{x_i}dt\leq \frac{2N\mathcal{E}(0)}{\kappa\psi_m}<\infty,\quad \forall~i,k\in\left\{1,\cdots,N \right\}. \]
	Since the primitive of $\sum_{ij}\|P_{ij}v_j- v_i\|_{x_i}^2$ satisfies all the conditions in  Barbalat's Lemma {(see \cite{B})}, we can conclude the desired estimate.
\end{proof}
\vspace{0.1cm}
\begin{remark} \label{R4.1}
In the sequel, we provide several comments on the results of Proposition \ref{P4.1}. \newline

\noindent 1. The first a priori condition in $\eqref{CSap}_1$ is related to modeling issue. By choosing $\psi$ to have a positive lower bound, we can get rid of the first a priori assumption, in contrast the second assumption is a genuine a priori condition to be checked for specific manifold. \\

\noindent 2.~Note that the estimates in Proposition \ref{P4.1} hold regardless of smooth structures and metrizations of given Riemannian manifold $(M,g)$. \\

\noindent 3.~Although $\mathcal{E}$ is smooth, nonnegative and monotonically decreasing, it does not imply the convergence of the derivative immediately. For instance, consider a following real-analytic function $h$:
\[h(x)=-\sum_{k=1}^{\infty}e^{-(k^2x-k^3)^2},\quad x\in\mathbb{R}. \]
Then, the function $h$ above does not converges to zero, while the primitive of $h$ is analytic, bounded and monotonically decreasing. Indeed, to obtain the zero convergence of $\frac{d\mathcal{E}}{dt}$, we need a uniform continuity of $\frac{d\mathcal{E}}{dt}$.\\

\noindent 4.~In \cite{H-K-S}, the authors verified the second \textit{a priori} condition $\eqref{CSap}_2$ for $M=\mathbb{S}^2$ and $M=\mathbb{H}^2$, and provided a sufficient framework to guarantee  $\eqref{CSap}_1$ for two-particle system on $\mathbb{S}^2$. Then, the authors in \cite{A-H-S} provided an explicit proof of  $\eqref{CSap}_2$ for $M=\mathbb{S}^d$ for every dimension $d$. In that sense, we here find an analogous proof of  $\eqref{CSap}_2$ for $M=\mathbb{H}^d$ any dimension $d$.
\end{remark}
\vspace{0.1cm}

\subsection{Velocity alignment of the HCS model} \label{sec:4.2}
In this subsection, we show that the second a priori condition $\eqref{CSap}_2$ can be verified for the HCS model by energy dissipation estimate and explicit form of parallel transport. 
First, we recall

\begin{equation}\label{HCS-1}
\begin{cases}
\displaystyle \dot{\bx}_i=\bv_i,\quad t>0,\quad  i=1,\cdots,N,\\
\displaystyle \dot{\bv}_i=||\bv_i||^2{\bx}_i+\frac{\kappa}{N}\sum_{j=1}^{N}\psi(\bx_{i}, \bx_{j} )\left({\bv_j}-{\bv_i}+\frac{\langle \bx_{i}, \bv_{j} \rangle_M }{1-\langle \bx_{i}, \bx_{j} \rangle_M }(\bx_{i}+\bx_{j})\right),\\
\displaystyle(\bx_i(0),\bv_i(0))=(\bx_i^{in},\bv_i^{in})\in \mathbb{R}^{d+1}\times \mathbb{R}^{d+1}, \\
\end{cases}
\end{equation}
subject to the initial constraints:
\begin{equation}\label{HCS-ini}
\langle \bx_i^{in},\bx_i^{in}\rangle_M =-1,\quad \langle \bx_i^{in},\bv_i^{in}\rangle_M =0,\quad i=1,\cdots,N. 
\end{equation}
Before we present velocity alignment estimate for \eqref{HCS-1} and \eqref{HCS-ini}, we begin with the following preparatory lemma. 
\vspace{0.2cm}

\begin{lemma}\label{L4.1}
	Let $\bx,\by$ and $\bv=(v^0,\cdots,v^d)$ be three vectors in $\mathbb{R}^{d+1}$ satisfying 
	\[\langle \bx,\bx\rangle_M =\langle \by,\by\rangle_M =-1\quad\mbox{and}\quad \langle \bx,\bv\rangle_M =0. \] 
	Then, we have the following inequality:
	\[\left|\frac{\langle \by,\bv\rangle_M }{1-\langle \bx,\by\rangle_M }\right|\leq \|\bv\|. \]
\end{lemma}
\begin{proof}
Consider a curve $x:\mathbb{R}\to \mathbb{H}^d$ and a point $y\in \mathbb{H}^d$ satisfying
	\begin{equation*}
	\iota(x(0))=\bx,\quad \iota(y)=\by,\quad d\iota_{x(0)}(\dot{x}(0))=\sum_{i=0}^{d}v^i\frac{\partial}{\partial x^i}\Big|_{\bx},
	\end{equation*}
	and denote $\bx(t):=\iota(x(t))$. Then, we use Proposition \ref{P2.4} and $\bv=\dot{\bx}(0)$ to obtain  
	\[\begin{aligned}
	\langle \by,\bv\rangle_M =\langle \by,\dot{\bx}(0)\rangle_M  &=\frac{d}{dt}\Big|_{t=0}\langle \by,\bx(t)\rangle_M  =-\frac{d}{dt}\Big|_{t=0}\cosh \Big(d(y,x(t))\Big)\\
	&=-\sinh\Big(d(y,x(0))\Big)\frac{d}{dt}\Big|_{t=0}d(y,x(t)).
	\end{aligned} \]
	Since $|\sinh(s)| \leq 1 + \cosh(s)$ and by using proposition 2.3, it suffices to show
	\begin{equation}\label{claim}
	\left|\frac{d}{dt}\Big|_{t=0}d(y,x(t))\right|\leq \|\bv\|.
	\end{equation}
	On the other hand, since geodesic distance $d(\cdot,\cdot)$ satisfies the triangle inequality, we have 
    \begin{equation}\label{speed}
    \left|d(y,x(t_1))-d(y,x(t_2))\right|\leq d(x(t_1),x(t_2))\leq \int_{t_1}^{t_2}\sqrt{g(\dot{x}(s),\dot{x}(s))}\mathrm {d}s=\int_{t_1}^{t_2}\|\dot{\bx}(s)\|\mathrm {d}s,
    \end{equation}
	for every $-\infty<t_1<t_2<\infty$. Hence, we divide \eqref{speed} by $t_2-t_1$ and take $t_1,t_2\to 0$ to obtain the claim \eqref{claim} and conclude the proof.
\end{proof}

\vspace{0.2cm}

Now, we are ready to show our second main result on the emergence of velocity alignment to the Cauchy problem to the HCS model \eqref{HCS-1} -- \eqref{HCS-ini}.  Our statement concerns a long-time behavior of solution to \eqref{model} in $\mathbb{H}^d$, but their analyses are indeed done for \eqref{hcs}--\eqref{hcsinitial}.
\vspace{0.2cm}
\begin{theorem} \label{T4.1}
Suppose that the coupling strength and communication weight satisfy 
\vspace{0.1cm}
\[  \kappa > 0, \quad  \inf_{(x, y) \in ({\mathbb H}^d)^2} \psi(x, y) \geq \psi_m, \]
for some positive constant $\psi_m$, and let $Z=(z_1,\cdots,z_N)$ be a solution to \eqref{model} on $\mathbb{H}^d$ with the initial data $\{(x_i^{in},v_i^{in})\}_{i=1}^N\subset T\mathbb{H}^d$. Then, the solution $Z=Z(t)$ exhibits asymptotic velocity alignment: 
	\[\lim_{t \to \infty} \|P_{ij} v_j- v_i\|_{x_i}=0,\quad \forall~ i,j=1,\cdots,N. \]
\end{theorem}
\begin{proof}
From Proposition \ref{P4.1}, it suffices to check that the a priori condition $\eqref{CSap}_2$ holds for system \eqref{HCS-1} -- \eqref{HCS-ini}: 
	\begin{equation}\label{uniform}
	\sup_{0 \leq t < \infty}\left|\frac{d}{dt}\|P_{ij} v_j- v_i\|_{x_i}^2 \right|<\infty,\quad \forall~i,j=1,\cdots,N.
	\end{equation}
	Then, we decompose $\frac{d}{dt}\|P_{ij} v_j- v_i\|_{x_i}^2$ into two terms and show their boundedness one by one. Namely, 
	\[\begin{aligned}
	\frac{d}{dt}\|P_{ij} v_j- v_i\|_{x_i}^2&=\frac{d}{dt}g(P_{ij} v_j-v_i,P_{ij}v_j-v_i)\\
	&=\frac{d}{dt}\left(\|v_i\|_{x_i}^2+\|v_j\|_{x_j}^2 \right)-2\frac{d}{dt}g(P_{ij}v_j,v_i)\\
	&=:\mathcal{I}_{11}+\mathcal{I}_{12}.
	\end{aligned} \]
	$\bullet$ Step A (Boundedness of $\mathcal{I}_{11}$): By the compatibility between metric tensor $g$ and Levi-Civita connection $\nabla$, we have 
	\[\begin{aligned}
	\mathcal{I}_{11}=2g(v_i,\nabla_{v_i} v_i)+2g(v_j,\nabla_{v_j} v_j),
	\end{aligned} \]
	where $v_i$ and $\nabla_{v_i} v_i$ are controlled by the initial kinetic energy. More precisely, since kinetic energy $\mathcal{E}$ is monotonically decreasing, we can apply Proposition \ref{P2.1} to obtain
	\begin{equation}\label{basicbounds}
	\begin{aligned}
	\|v_i(t)\|_{x_i(t)}&\leq \sqrt{2\mathcal{E}[t]}\leq \sqrt{2\mathcal{E}[0]},\quad \forall~i=1,\cdots,N,\\
	\|\nabla_{v_i(t)} v_i(t)\|_{x_i(t)}&\leq \frac{\kappa}{N}\sum_{k=1}^{N}\psi(x_i(t),x_k(t))\|P_{ik}v_k(t)-v_i(t)\|_{x_i}\leq 2\kappa\psi_M\sqrt{2\mathcal{E}[0]}.
	\end{aligned}
	\end{equation}
	Therefore, $\mathcal{I}_{11}$ has a {finite uniform upper bound} in $t$.\\
	\vspace{0.2cm}
	
	\noindent $\bullet$ Step B (Boundedness of $\mathcal{I}_{12}$): For solution $Z=(\bz_1,\cdots,\bz_N)$, we consider a $(d+1)$-dimensional representations of $\bx_1,\cdots,\bx_N$ by using \eqref{tangent}. Then, we combine Proposition \ref{P2.3} and \ref{P2.4} to obtain an explicit formula for $g(P_{ij}\bv_j,\bv_i)$:
	\begin{equation}\label{inner}
	g(P_{ij}v_j,v_i) =\left\langle \bv_j+\frac{\langle \bv_j,\bx_i\rangle_M }{1-\langle \bx_i,\bx_j\rangle_M }(\bx_i+\bx_j),\bv_i\right\rangle_M  =\langle\bv_i,\bv_j\rangle_M +\frac{\langle \bv_i,\bx_j\rangle_M \langle \bv_j,\bx_i\rangle_M }{1-\langle \bx_i,\bx_j\rangle_M }.
	\end{equation},where we used $\langle \bx_i,\bv_i \rangle_M $$=0$.
	Now, we decompose $t$-derivative of \eqref{inner} into 
	\begin{equation}
	\begin{aligned}
	\frac{d}{dt}g(P_{ij}v_j,v_i)&=\left(\langle\dot{\bv}_i,\bv_j\rangle_M +\frac{\langle \dot{\bv}_i,\bx_j\rangle_M \langle \bv_j,\bx_i\rangle_M }{1-\langle \bx_i,\bx_j\rangle_M }\right)+\left(\langle\bv_i,\dot{\bv}_j\rangle_M +\frac{\langle \bv_i,\bx_j\rangle_M \langle \dot{\bv}_j,\bx_i\rangle_M }{1-\langle \bx_i,\bx_j\rangle_M }\right)\\
	&\hspace{0.5cm}+\frac{\langle \bv_i,\bx_j\rangle_M +\langle {\bv}_j,\bx_i\rangle_M }{1-\langle \bx_i,\bx_j\rangle_M }\left(\langle \bv_i,\bv_j\rangle_M +\frac{\langle \bv_i,\bx_j\rangle_M \langle \bv_j,\bx_i\rangle_M }{1-\langle \bx_i,\bx_j\rangle_M } \right)\\
	&=:\mathcal{I}_{121}+\mathcal{I}_{122}+\mathcal{I}_{123}.
	\end{aligned}
	\end{equation}
	Below, we further provide detailed estimates on $\mathcal{I}_{12i}$ for each $i=1,2,3$ separately.\\
	
	\noindent $\circ$ Case A (Estimates on $\mathcal{I}_{121}$ and $\mathcal{I}_{122}$):~First, we employ Proposition \ref{P2.4} to replace $\dot{\bv}_i$ by $\nabla_{\bv_i}\bv_i+\|\bv_i\|^2\bx_i$. Then, we apply Lemma \ref{L4.1} to simplify $\mathcal{I}_{121}$ as
	\[\begin{aligned}
	\mathcal{I}_{121}&=\langle\dot{\bv}_i,\bv_j\rangle_M +\frac{\langle \dot{\bv}_i,\bx_j\rangle_M \langle \bv_j,\bx_i\rangle_M }{1-\langle \bx_i,\bx_j\rangle_M }\\
	&=\left(\langle\nabla_{\bv_i}\bv_i,\bv_j\rangle_M +\frac{\langle \nabla_{\bv_i}\bv_i,\bx_j\rangle_M \langle \bv_j,\bx_i\rangle_M }{1-\langle \bx_i,\bx_j\rangle_M }\right)+\left(\langle \|\bv_i\|^2\bx_i,\bv_j\rangle_M +\frac{\langle \|\bv_i\|^2\bx_i,\bx_j\rangle_M \langle \bv_j,\bx_i\rangle_M }{1-\langle \bx_i,\bx_j\rangle_M }\right)\\
	&=\left\langle \nabla_{\bv_i}\bv_i,\bv_j+\frac{\langle \bv_j,\bx_i\rangle_M }{1-\langle\bx_i,\bx_j\rangle_M }(\bx_i+\bx_j)\right\rangle_M +\|\bv_i\|^2\frac{\langle \bv_j,\bx_i\rangle_M }{1-\langle\bx_i,\bx_j\rangle_M }\\
	&=g(\nabla_{v_i} v_i,P_{ij}v_j) +\|\bv_i\|^2\frac{\langle \bv_j,\bx_i\rangle_M }{1-\langle\bx_i,\bx_j\rangle_M }.
	\end{aligned}\]
	Therefore, $\mathcal{I}_{121}$ can be controlled by a function of initial kinetic energy:
	\[|\mathcal{I}_{121}|\leq \|\nabla_{v_i}v_i\|_{x_i}\|v_j \|_{x_j}+\|v_i\|_{x_i}^2\|v_j\|_{x_j}\leq 4\kappa\psi_M\mathcal{E}[0]+(2\mathcal{E}[0])^{\frac{3}{2}}<\infty, \] 
    and we get the same upper bound for $|\mathcal{I}_{122}|$ by using similar argument as above.\\
    
    \noindent $\circ$ Case B (Estimate on $\mathcal{I}_{123}$):~ In this case, we substitute the relation \eqref{inner} to $\mathcal{I}_{23}$ and apply Lemma \ref{L4.1} to deduce 
    \[|\mathcal{I}_{123}| \leq \left(\|v_i\|_{x_i}+\|v_j\|_{x_j}\right)\left|g(P_{ij} v_j,v_i) \right|\leq 2(2\mathcal{E}[0])^{\frac{3}{2}}<\infty. \]
Finally, we combine the uniform boundedness of $\mathcal{I}_{11}$ and $\mathcal{I}_{12}$ to obtain \eqref{uniform} and conclude the desired velocity alignment.
\end{proof}

\subsection{A dichotomy in asymptotic patterns} \label{sec:4.3}
In this subsection, we characterize an asymptotic pattern of spatial configuration $(x_1,\cdots,x_N)$, when the asymptotic flocking \eqref{flocking} emerges on $\mathbb{H}^2$. We first introduce a preparatory lemma to estimate the turning angle of tangent vector by parallel transports along a geodesic triangle on $\mathbb{H}^2$. This estimate was first introduced in \cite{A-H-S} for the Cucker-Smale model on $\mathbb{S}^2$ by using the Gauss-Bonnet theorem for geodesic triangle. 

\begin{lemma}\label{L4.2} \emph{\cite{A-H-S}}
	Let $x_1,x_2$ and $x_3$ be three points in $\mathbb{H}^2$ and $v \in T_{x_1}\mathbb{H}^2$ be a tangent vector at point $x_1$, and let $\Delta = \Delta(x_1, x_2, x_3)$ be a geodesic triangle whose edges consist of (length-minimizing) geodesics between each pair of points $\left\{x_1,x_2,x_3 \right\}$.
	Then, we have 
	\begin{equation}\label{Holonomy}	\|P_{12}P_{23}P_{31}v-v\|_{x_1}=\|v\|_{x_1}\sqrt{2-2\cos(\emph{Area}(\Delta))}.
	\end{equation} 
\end{lemma}
\begin{proof}
This is a direct consequence of Gauss-Bonnet theorem for compact surface with boundary, and the proof of the analogous result for spherical triangle are provided in \cite{A-H-S}. The only difference here is that the sectional curvature $K$ is fixed as $-1$ and therefore the Gauss-Bonnet theorem says
\[-\mbox{Area}(T)=2\pi-(\alpha_1+\alpha_2+\alpha_3), \] 
where $\alpha_i$ denotes the signed exterior angle of $T$ at vertex $\bx_i$. Since the angle between $v$ and $P_{12}P_{23}P_{31}v$ is congruent to $-(\alpha_1+\alpha_2+\alpha_3)$ modulo $2\pi$ for nonzero $\bv$, we conclude the desired result \eqref{Holonomy}.
\end{proof}
Now, we recall an inequality showing the relation between relative velocities and areas of geodesic triangles introduced in \cite{A-H-S}.
\begin{lemma}\label{L4.3}\emph{\cite{A-H-S}}
	Let $(M,g)$ be an arbitrary complete Riemannian manifold and 
	\[(x_1,v_1),\cdots,(x_N,v_N)\in TM. \]  For each pair $(i,j)\in\left\{1,\cdots,N \right\}^2$, we fix a geodesic $\gamma_{ij}$ on $M$ that joins from $x_j$ to $x_i$, and {let} $P_{ij}$ be the corresponding parallel transport operator along $\gamma_{ij}$. Then, we have 
	\[\sum_{i, j,k} \|P_{ki}P_{ij}P_{jk}v_k-v_k\|_{x_k}^2\leq 9N\sum_{i, j}\|P_{ij}v_j-v_i\|_{x_i}^2. \]
\end{lemma}
\begin{proof}
	Although the statement of Lemma \ref{L4.2} can be written under more general setting, the proof is completely the same as in \cite{A-H-S}. Nevertheless, we here introduce a proof for the completeness.	Since each parallel transport $P_{ij}$ acts as an isometry between tangent space $T_{x_i}M$ and $T_{x_j}M$, we have 
	\begin{equation}\label{tri}
	\begin{aligned}
	&\|P_{ki}P_{ij}P_{jk}v_k- v_k\|_{x_k} \\
	&\hspace{0.5cm} \leq \|P_{ki}P_{ij}P_{jk} v_k-P_{ki}P_{ij} v_j\|_{x_k}+\|P_{ki}P_{ij} v_j-P_{ki} v_i\|_{x_k}+\|P_{ki} v_i-v_k\|_{x_k}\\
	& \hspace{0.5cm} =\|P_{jk} v_k-v_j\|_{x_j}+\|P_{ij}v_j-v_i\|_{x_i}+\|P_{ki}v_i-v_k\|_{x_k}.
	\end{aligned}
	\end{equation}
	Therefore, we take a sum of squares of \eqref{tri} for all $i,j,k$ to obtain
	\begin{align*}
	\begin{aligned}
	\sum_{i, j,k}\|P_{ki}P_{ij}P_{jk} v_k-v_k\|_{x_k}^2 &\leq 3\sum_{i, j,k} \left(\|P_{jk}v_k-v_j\|_{x_j}^2+\|P_{ij} v_j-v_i\|_{x_i}^2+\|P_{ki}v_i-v_k\|_{x_k}^2\right)\\
	&=9N\sum_{i, j}\|P_{ij} v_j-v_i\|_{x_i}^2.
	\end{aligned}
	\end{align*}
\end{proof}
In \cite{A-H-S}, the authors introduced two well-known properties on spherical trigonometry without proofs to show dichotomy in asymptotic patterns of CS particles. Here we present their corresponding results for the hyperbolic triangles on $\mathbb{H}^2$ to figure out the detailed asymptotic patterns of CS particles.
\begin{lemma}\label{L4.4} Let $\Delta$ be a hyperbolic geodesic triangle on $\mathbb{H}^2$, and we consider the hyperboloid model as a realization of $\mathbb{H}^2$. Denote $A,B$ and $C$ as three vertices of $\Delta$, and define the angles $b,c$ and $\angle A$ as 
	\[b=d(C,A),~c=d(A,B),~\angle A: \emph{interior angle of}~\Delta ~\emph{at vertex}~A.  \]
	Then, the following assertions hold:
	\begin{enumerate}
				\item (L'Huilier's formula for hyperboloid): 
			\[\tan\left(\frac{\emph{Area}(T)}{2}\right)=\frac{\tanh\displaystyle\frac{b}{2}\tanh\frac{c}{2}\sin \angle A}{1-\tanh\displaystyle\frac{b}{2}\tanh\frac{c}{2}\cos \angle A}. \]
				\item (Law of sines):
				\[\left|\sinh b\sinh c\sin \angle A\right|=\left|\det(OA,OB,OC)\right|. \]
			\end{enumerate}
\end{lemma}
\begin{proof}
Since a proof is very lengthy and technical, we leave it in Appendix \ref{App-C}.
\end{proof}

Now, we are ready to present our third main result{, where the precise statement is given} as follows.

\begin{theorem}\label{T4.2}
	Let $Z=(x_1,\cdots, x_N)$ be a global smooth solution to \eqref{model} exhibiting asymptotic flocking \eqref{flocking}.
	Then, we have the following dichotomy:
	\begin{enumerate}
		\item The energy tends to zero:
		\[\lim_{t\to\infty}\mathcal{E}[t]=0.\]
		\item The energy converges to a positive value, and position configuration becomes coplanar asymptotically: 
		\[\lim_{t\to\infty}\mathcal{E}[t]>0,\quad \lim_{t \to \infty}\det\left(x_i(t)\Big|x_j(t)\Big|x_k(t)\right)=0,\qquad \forall~ i,j,k=1,\cdots,N. \]
	\end{enumerate}
\end{theorem}
\begin{proof}
	It follows from Proposition \ref{P4.1} that the energy always converges to a nonnegative value:
	\begin{equation}
	\exists~\mathcal{E}^\infty:=\lim_{t \to \infty}\mathcal{E}[t]\geq 0.
	\end{equation}
	Therefore, it suffices to show that
	\[\lim_{t \to \infty}\det(x_i(t)|x_j(t)|x_k(t))=0\quad \forall~ i,j,k=1,\cdots,N \]
	holds, whenever the asymptotic flocking \eqref{flocking} emerges and $\mathcal{E}^\infty$ is strictly positive. In this case, the speed of each particles converges to a common positive number $\sqrt{\frac{\mathcal{E}^\infty}{N}}$, since 
	\[\left|\|v_i\|_{x_i}-\|v_j\|_{x_j}\right|= \left|\|v_i\|_{x_i}-\|P_{ij} v_j\|_{x_i}\right|\leq \|P_{ij} v_j- v_i\|_{x_i}\to 0,\quad \forall ~i,j=1,\cdots,N. \]
	Now, we apply Lemma \ref{L4.2} and Lemma \ref{L4.3} to \eqref{flocking} to obtain 
	\[\lim_{t \to \infty} \sum_{i, j,k} \Big(2-2\cos(A_{ijk}(t)) \Big)=0, \]
	where $A_{ijk}(t)$ denotes the (length-minimizing) geodesic triangle with vertices $\{x_i,x_j,x_k \}$ at time $t$. Hence, we deduce $A_{ijk}(t)\to 0$ for each $i,j,k$, and we use Lemma \ref{L4.4} to obtain
	\[\lim_{t \to \infty}\tanh\bigg(\frac{d(x_i,x_j)}{2}\bigg)\tanh\bigg(\frac{d(x_i,x_k)}{2}\bigg)\sin\angle(x_j x_i x_k)\to 0,\quad \forall ~i,j=1,\cdots,N. \]
	Finally, we combine the above result with the boundedness of $d(x_i,x_j)$ and $d(x_i,x_k)$ from \eqref{flocking} to conclude the desired result:
	\[\left|\det\left(x_i(t)\Big|x_j(t)\Big|x_k(t)\right)\right|=\left|\sinh\bigg(\frac{d(x_i,x_j)}{2}\bigg)\sinh\bigg(\frac{d(x_i,x_k)}{2}\bigg)\sin\angle(x_j x_i x_k)\right|\to 0\quad \mbox{as}~~t\to 0. \]
\end{proof}

\subsection{Numerical simulation} \label{sec:4.4}

In this subsection, we present a numerical example of the HCS model on $\mathbb{H}^{2}$ for $N=10$ and compare them with our analytical results in Section \ref{sec:4}. %, which illustrate the behavior of particles, when asymptotic velocity flocking emerges, and two graphs of kinetic energy of dynamics and determinant of first, second and third particles. 
We used the fourth-order Runge-Kutta method in all numerical simulations, {and we orthogonally projected $(x_i(n+1) ,v_i(n+1))$ to $T\mathbb{H}^d$ after each step.} Moreover, we fixed the following parameters: 
\[ \phi(x,y)=1,\quad N=10, \quad \kappa=1, \quad dt=0.001, \quad t = 200s. \]
%For visual convenience, we properly rotate the x axis by right angle i.e $(x,y,z)\rightarrow (y,z,x)$, and for  the numerical integration of the HCS model on $\mathbb{H}^{1}$, we set 
%\[ \phi(r)=\textcolor{red}{\cos}(r), \quad \kappa=5, \quad N = 10, \quad dt=0.01, \quad t = 200s. \]
%\subsection{N=10}  
Recall that according to Theorem \ref{T4.1} and Theorem \ref{T4.2}, we obtained the followings:

\begin{enumerate}
	\item Energy $\mathcal{E}[t]$ decrease monotonically.
	\item The velocity alignment always occur.
	\item If all distances between particles are bounded in $t$, then they approach to a geodesic of $\mathbb{H}^2$ unless $\mathcal{E}[t]\to 0$.
\end{enumerate}

\noindent We here show a $t$ vs $\log\mathcal{E}[t]$ plot and $t$ vs $\det(x_1|x_2|x_3)$ plot to see (1) and (3).
\begin{figure}[!h]
	\begin{center}
		\begin{subfigure}[b]{0.48\linewidth}
			{\includegraphics[width=1.0\textwidth]{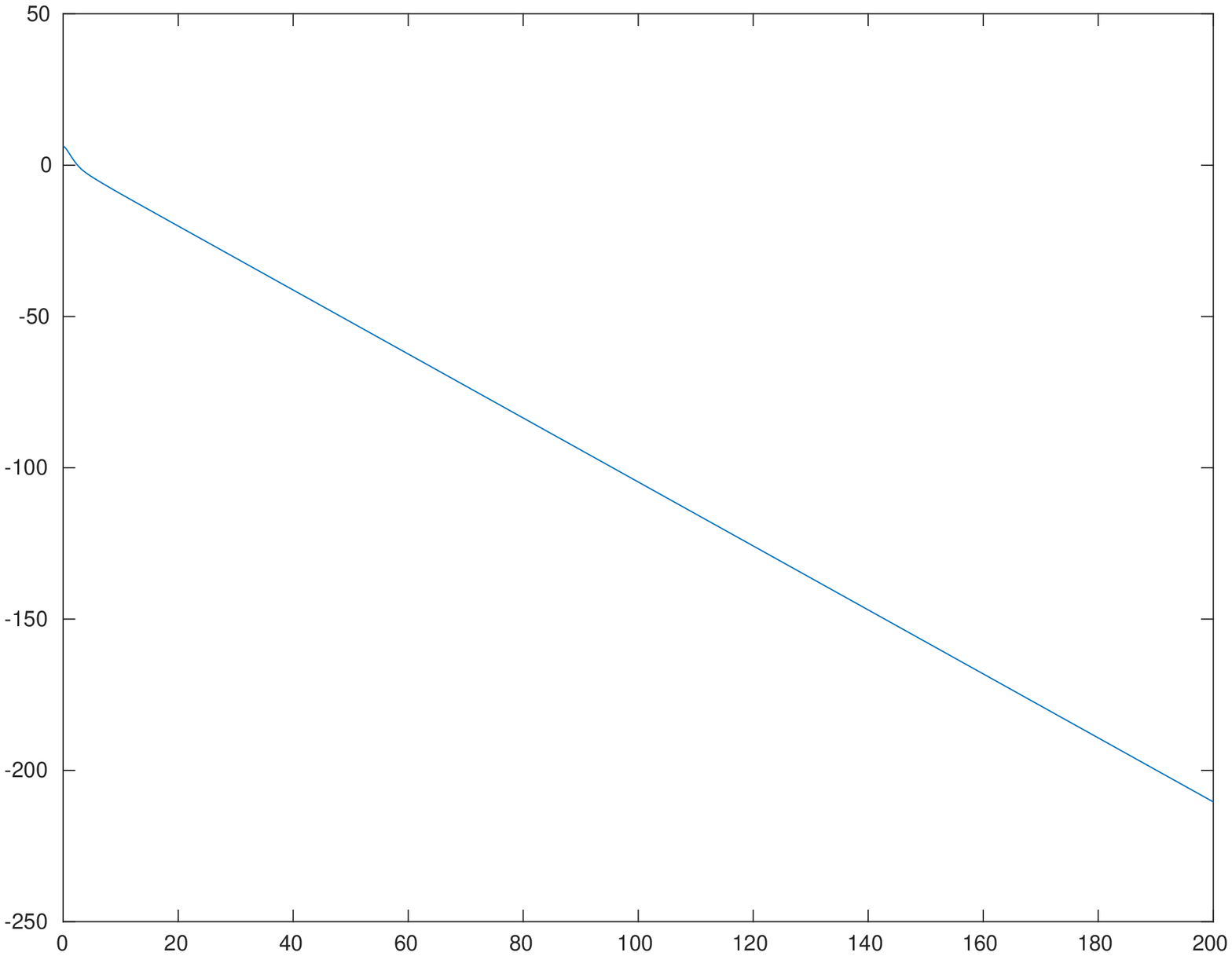}}
			\caption{Evolution of $\log\mathcal{E}[t]$.}
		\end{subfigure}
		\begin{subfigure}[b]{0.48\linewidth}
			{\includegraphics[width=1.0\textwidth]{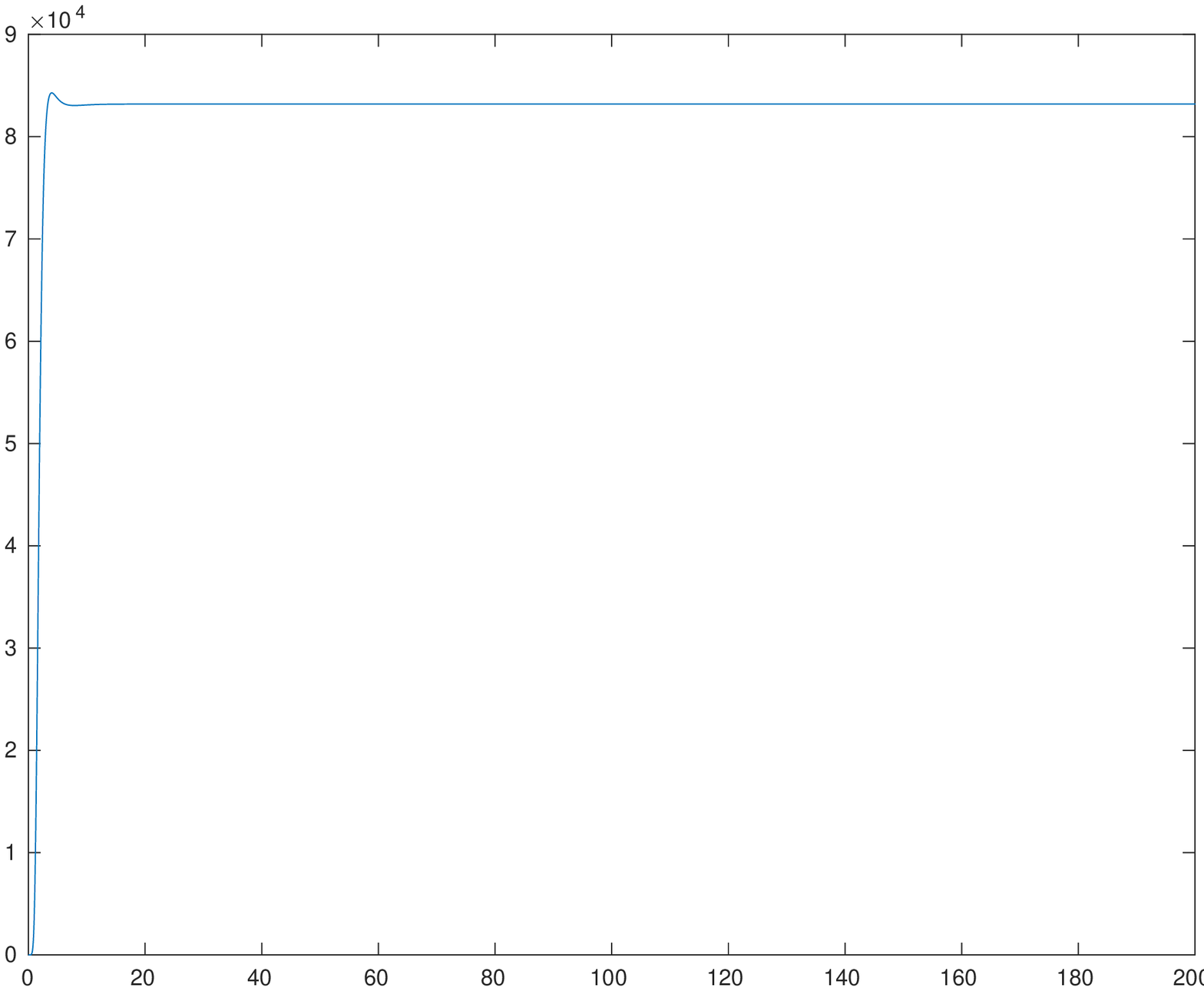}}
			\caption{Evolution of $\det(x_1|x_2|x_3)$.}
		\end{subfigure}
	\end{center}
	\figurename{ 1: Temporal evolution of $\log\mathcal{E}[t]$ and $\det(x_1|x_2|x_3)$. }
\end{figure}
%\begin{figure}[!h]
%	\begin{center}
%		\begin{subfigure}[b]{1.0\linewidth}
%			{\includegraphics[width=1.0\textwidth]{loge.eps}}
%			%\caption{Temporal evolution of $\log\mathcal{E}[t]$}
%		\end{subfigure}
%	\end{center}
%\figurename{ 1: Temporal evolution of $\log\mathcal{E}[t]$}
%	\end{figure}
%\begin{figure}[!h]
%	\begin{center}
%		\begin{subfigure}[b]{1.0\linewidth}
%			{\includegraphics[width=1.0\textwidth]{det.eps}}
%			%\caption{Temporal evolution of $\det(x_1|x_2|x_3)$}
%		\end{subfigure}
%	\end{center}
%	\figurename{ 2: Temporal evolution of $\det(x_1|x_2|x_3)$. }
%\end{figure}\\

Fig. 1(A) shows that the energy indeed converges to zero exponentially fast, and therfore each particle converged to a certain point which might not be located in the same geodesic. Since all particles converge, the determinant of first three particles also converges to a nonzero value as in Fig. 1(B).\\

{One considerable issue here} is that we could not find any other type of result in our countless trial {for random initial data}. In our simulations, the energy always converges to zero exponentially fast, and therefore the particles could not approach {further} until they aggregate to a single geodesic. Even if they start from any geodesic $\gamma$ with initial velocity slightly different from $\dot{\gamma}$, they show a similar dynamics. This indicates that the {coplanar} state in our dichotomy is indeed unstable and cannot emerges from generic initial data/system parameter. This is the contrasted difference with the unit $d$-sphere $\mathbb{S}^d$ in \cite{A-H-S}, since there was a flocking on $\mathbb{S}^d$ with nonzero asymptotic velocity, and therefore all particles approach to a single great circle asymptotically.  {Moreover, from our derivation of HK from HCS and the asymptotic stability results in \cite{H-K-R}, it might be possible to achieve asymptotic flocking with nonzero $\mathcal{E}^\infty$ if there is a common invariant geodesic along the flow. In our simulation on a geodesic, the energy $\mathcal{E} $ converged to a nonzero positive value, and the determinant remained zero for short time (until $t=50s$) because of numerical error and possibly the instability of coplanar state.}

\begin{figure}[!h]
	\begin{center}
		\begin{subfigure}[b]{0.48\linewidth}
			{\includegraphics[width=1.0\textwidth]{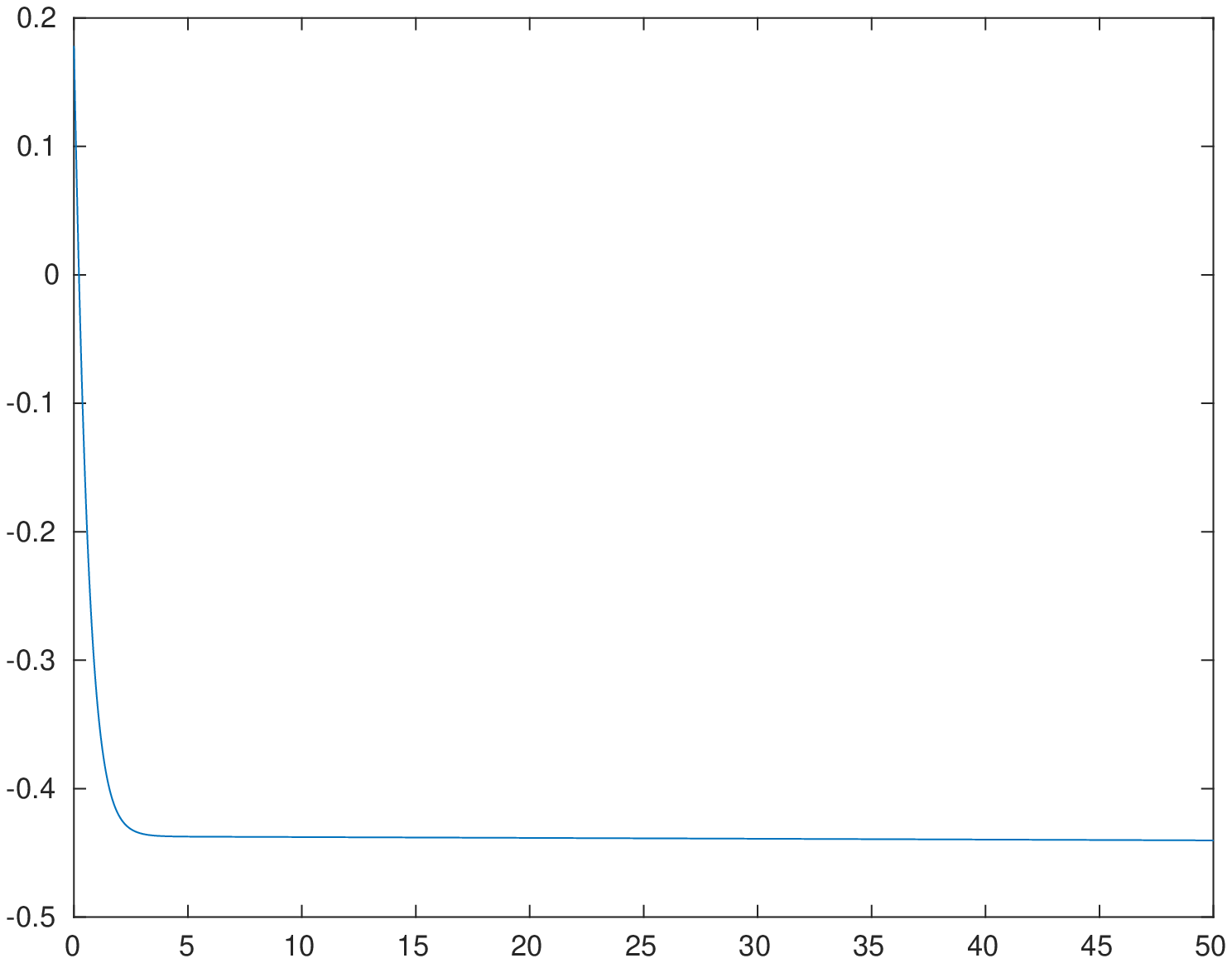}}
			\caption{Evolution of $\log\mathcal{E}[t]$ for HK.}
		\end{subfigure}
		\begin{subfigure}[b]{0.48\linewidth}
			{\includegraphics[width=1.0\textwidth]{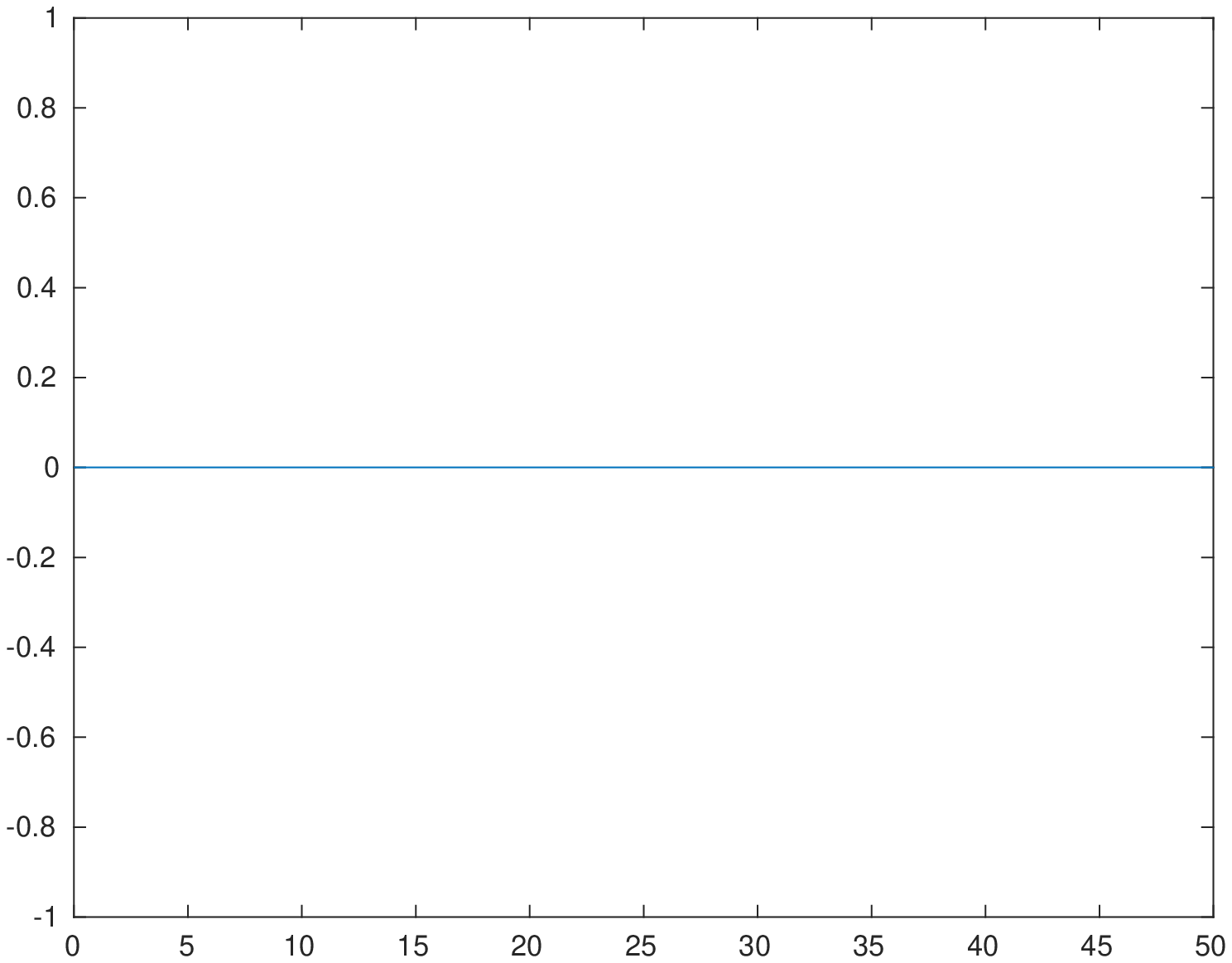}}
			\caption{Evolution of $\det(x_1|x_2|x_3)$ for HK.}
		\end{subfigure}
	\end{center}
	\figurename{ 3: Temporal evolution of $\log\mathcal{E}[t]$ and $\det(x_1|x_2|x_3)$ for HK.}
\end{figure}

\section{Conclusion} \label{sec:5}
\setcounter{equation}{0}
In this paper, we have studied emergent behaviors of the Cucker-Smale model on $\mathbb{H}^d$( in short the HCS model) whose general abstract form was already introduced in \cite{H-K-S} for generic complete Riemannian manifolds. For ${\mathbb H}^d$, we can explicitly express geodesics and parallel transport in terms of state variables. Thus, the explicit form of the HCS model can be derived. Thanks to this explicit representation of the HCS model, three ramifications were obtained in this paper. First, when all initial positions are restricted to a geodesic and initial velocities are tangent to the geodesic, geodesic becomes a positively invariant set for the HCS dynamics. In this case, we showed that the dynamics of the HCS model reduces to that of the hyperbolic Kuramoto model which was proposed in a recent work \cite{R-L-W}. Second, any pair of points on ${\mathbb H}^d$ admits a unique geodesic and thus the choice of communication weight $\psi$ can be done without any consideration of conjugate points as in $\mathbb{S}^d$. Then, for a bounded symmetric smooth communication function $\psi(\cdot,\cdot)$ with a positive lower bound, we were able to show the emergence of velocity alignment {\it without a priori condition,} which improves earlier a priori velocity alignment result in \cite{H-K-S}. Third, we also found a dichotomy on the asymptotic patterns of CS particles on ${\mathbb H}^2$. Thanks to dissipative structure of the velocity interaction terms, the total kinetic energy is monotonically decreasing and bounded below by zero, hence it converges to a nonnegative value asymptotically. When flocking occurs, there will be two possible asymptotic scenario, either the total kinetic energy converges to zero or the total energy converge a positive value and spatial positions lie on the same plane in ${\mathbb R}^3 \cap {\mathbb H}^2$ asymptotically.  \newline

There are several open remaining that were not investigated further in this paper. For example, in our velocity alignment estimate, we have assumed that the communication weight function does have a positive lower bound, so whether we can get rid of this positive lower bound condition or not will be an interesting question. Another interesting question will be identification of basins leading to the dichotomy. These issues will be addressed in a future work. 

\appendix

\newpage

\section{Proof of Proposition \ref{P2.1}} \label{App-A}
\setcounter{equation}{0}
In this appendix, we discuss the completeness of $(\mathbb{H}^d,g)$ and recover the well-known simple formula for geodesics on the hyperboloid model which provides a proof for Proposition \ref{P2.1}. \newline

First, for the canonical orthonomal basis $\left\{\frac{\partial}{\partial x^1}\Big|_{u},\cdots,\frac{\partial}{\partial x^d}\Big|_{u} \right\}$ of tangent space $T_u\mathbb{R}^d$ at the point $u=(u^1,\cdots,u^d)$, the metric tensor $g$ is designated as follows: for $i,j=1,\cdots,N,$
\vspace{0.2cm}
\begin{equation*}
\begin{aligned}
g_{ij}(u)&:=g\left(d\phi^{-1}\left(\frac{\partial}{\partial x^i}\Big|_{u}\right),~ d\phi^{-1}\left(\frac{\partial}{\partial x^j}\Big|_{u}\right) \right)=\delta_{ij}-\frac{u^iu^j}{1+\sum_m(u^m)^2},%\\
%g^{ij}(u)&=\delta_{ij}+u^iu^j,\quad  (g^{ij}):=(g_{ij})^{-1}.
\end{aligned}
\end{equation*}
\vspace{0.2cm}

\noindent and this designation uniquely determines an inner product in the tangent space $T_{\phi^{-1}(u)}\mathbb{H}^d$, since 
\begin{align}
\begin{aligned}\label{metric}
&g\left(d\phi^{-1}\left(a^i\frac{\partial}{\partial x^i}\Big|_{u}\right),~ d\phi^{-1}\left(a^j\frac{\partial}{\partial x^j}\Big|_{u}\right) \right) \\
&\hspace{1cm} = \sum_m(a^m)^2-\frac{\left(\sum_i a^iu^i\right)^2}{1+\sum_m(u^m)^2} \geq \sum_m(a^m)^2\left(1-\frac{\sum_m(u^m)^2}{1+\sum_m(u^m)^2}\right) \geq 0,
\end{aligned}
\end{align}
for every~$(a^1,\cdots,a^d)~\mbox{in}~ \mathbb{R}^d.$

In addition, the inverse matrix of $(g_{ij})$ is then given by
\begin{equation}\label{dual}
g^{ij}(u)=\delta_{ij}+u^iu^j,\quad  (g^{ij}):=(g_{ij})^{-1},
\end{equation}
which can be verified from the following simple calculations:
\[\begin{aligned}
g_{ij}g^{jk}&=\left(\delta_{ij}-\frac{u^iu^j}{1+\sum_m(u^m)^2} \right)\left(\delta_{jk}+u^ju^k \right)\\
&=\delta_{ik}+u^iu^k-\frac{u^iu^k}{1+\sum_m(u^m)^2}-\frac{\sum_m(u^m)^2}{1+\sum_m(u^m)^2}u^iu^k =\delta_{ik}.
\end{aligned} \]

\noindent Recall that for the Levi-Civita connection $\nabla$ compatible with the metric tensor $g$, the Christoffel symbols of the second kind $\Gamma_{ij}^k$ are given as
\begin{equation*}
\begin{aligned}
\Gamma_{ij}^k(u)=\frac{1}{2}g^{kl}(u)\left(\frac{\partial g_{li}}{\partial u^j}+\frac{\partial g_{jl}}{\partial u^i}-\frac{\partial g_{ij}}{\partial u^l} \right),
\end{aligned}
\end{equation*}
where 
\[\nabla_{\partial_i}\partial_j=\Gamma_{ij}^k\partial_k,\quad \partial_i:= d\phi^{-1}\left(\frac{\partial}{\partial x^i}\Big|_{u}\right).\]
More explicitly, we obtain
\begin{align*}
\frac{\partial}{\partial u^b}g_{ca}=\frac{\partial}{\partial u^b}\left(\delta_{ca}-\frac{u^c u^a}{1+\sum_m(u^m)^2}\right)=-\frac{\delta_{bc} u^a+u^c\delta_{ab}}{1+\sum_m(u^m)^2}+\frac{2u^cu^au^b}{\left(1+\sum_m(u^m)^2\right)^2}
\end{align*}
so that the detailed formula for $\Gamma_{ij}^k$ follows immediately:
\begin{equation}\label{Christoffel}
\begin{aligned}
\Gamma_{ij}^k(u)=&\frac{1}{2}g^{kl}(u)\left(\frac{\partial g_{li}}{\partial u^j}+\frac{\partial g_{jl}}{\partial u^i}-\frac{\partial g_{ij}}{\partial u^l} \right)\\
=&g^{kl}(u)\left(-\frac{u^l\delta_{ij}}{1+\sum_m(u^m)^2}+\frac{u^iu^ju^l}{\left(1+\sum_m(u^m)^2\right)^2}\right)\\
=&-\frac{g^{kl}(u)u^l}{1+\sum_m(u^m)^2}g_{ij}(u)\\
=&-u^kg_{ij}(u),\quad \forall~i,j,k=1,\cdots,N.
\end{aligned}
\end{equation}
Therefore, for every geodesic $\gamma(s)=\phi^{-1}(u^1(s),\cdots,u^d(s))$ of $(\mathbb{H}^d,g)$, we have
\begin{equation}\label{geod0}
\frac{d^2u^k}{ds^2}+\Gamma^k_{ij}\frac{du^i}{ds}\frac{du^j}{ds}=\ddot u^k-u^kg_{ij}(u)\dot u^i\dot u^j=\ddot u^k-u^kg(\dot\gamma,\dot\gamma)=0,\quad k=1,\cdots,d,
\end{equation}
where $g(\dot\gamma,\dot\gamma)$ is constant along each geodesic curve. From the uniqueness of ODE theory, we obtain 
\begin{equation}\label{geodesic1}
u(s)=\phi(\gamma(s))=u(0)\cosh cs+\frac{\dot u(0)}{c}\sinh cs,\quad c^2:=g(\dot\gamma,\dot\gamma).
\end{equation}
Finally, we consider the geodesic $\gamma(s)$ as a curve in $\mathbb{R}^{d+1}$:
\begin{align*}
\begin{aligned}
\gamma(s) &=\phi^{-1}(u^1(s), \cdots, u^d(s)) \\
&=\left(\sqrt{1+\sum_m(u^m)^2}, u^1, u^2, \cdots, u^d\right) =:\left(u^0(s), u^1(s), u^2(s), \cdots, u^d(s)\right).
\end{aligned}
\end{align*}
Then, we combine \eqref{geodesic1} and 
\[c^2=g(\dot\gamma,\dot\gamma)=\sum_m(\dot u^m(0))^2-\frac{\left(\sum_mu^m(0)\dot u^m(0)\right)^2}{1+\sum_m(u^m(0))^2} \]
to compute $u^0$ explicitly. Precisely, one has
\[\begin{aligned}
u^0(s)^2&=1+\sum_m(u^m(s))^2\\
&=1+\sum_m\left(u^m(0)\cosh cs+\frac{\dot u^m(0)}{c}\sinh cs\right)^2\\
&=1+\sum_m\left[\left(u^m(0)\right)^2\cosh^2 cs+\frac{2u^m(0)\dot u^m(0)}{c}\cosh cs\sinh cs+\frac{(\dot u^m(0))^2}{c^2}\sinh^2 cs\right]\\
&=1+\sum_m\left[\left(u^m(0)\right)^2\cosh^2 cs+\frac{2u^m(0)\dot u^m(0)}{c}\cosh cs\sinh cs\right]\\
&\hspace{0.7cm}+\sinh^2 cs+\frac{\left(\sum_mu^m(0)\dot u^m(0)\right)^2}{c^2(1+\sum_m(u^m(0))^2)}\sinh^2 cs
\end{aligned} \]
\[\begin{aligned}
&=\left(1+\sum_m \left(u^m(0)\right)^2\right)\cosh^2 cs+\frac{2}{c}\left(\sum_m u^m(0)\dot u^m(0)\right)\cosh cs\sinh cs\qquad\\
&\hspace{0.7cm}+\frac{\left(\sum_mu^m(0)\dot u^m(0)\right)^2}{c^2(1+\sum_m(u^m(0))^2)}\sinh^2 cs\\
&=\left[\sqrt{1+\sum_m \left(u^m(0)\right)^2}\cosh cs+\frac{\left(\sum_mu^m(0)\dot u^m(0)\right)}{c\sqrt{1+\sum_m \left(u^m(0)\right)^2}}\sinh cs\right]^2\\
&=\left[u^0(0)\cosh cs+\frac{\dot u^0(0)}{c}\sinh cs \right]^2.
\end{aligned} \]

\bigskip

\section{Proof of Proposition \ref{P2.4}} \label{App-B}
\setcounter{equation}{0}
\noindent In this appendix, we provide proofs for two assertions stated in Proposition \ref{P2.4}. \newline

\noindent $\bullet$~(Proof of the first assertion): First of all, we combine \eqref{geod0} and Proposition $\ref{P2.3}$ to obtain
		\begin{equation}\label{covar}
		a^i=\ddot{x}^i-g(\dot x,\dot x)x^i=\ddot{x}^i-\langle
		\dot{\bx},\dot{\bx}\rangle_M  x^i,\quad i=1,\cdots,d.
		\end{equation}
		On the other hand, by taking derivative twice in Proposition \ref{P2.3}, we have 
		\begin{equation}\label{ddot}
		\langle\bx,\ddot{\bx}\rangle_M +\langle\dot{\bx},\dot{\bx}\rangle_M \equiv 0.
		\end{equation}
		Now, we substitute the relation \eqref{ddot} to \eqref{covar} to get 
		\begin{equation}\label{Projection}
		a^i=\ddot{x}^i+\langle\bx,\ddot{\bx}\rangle_M  x^i,\quad i=1,\cdots,d.
		\end{equation}
		Finally, we use the uniqueness of $(d+1)$-dimensional vector $\mathbf{a}=(a^0,\cdots,a^d)$ satisfying \eqref{Projection} and $\langle\mathbf{a},\bx\rangle_M =0$ to conclude \eqref{covariant}.\\
		
		\noindent $\bullet$~(Proof of the second assertion):~It follows from the estimate \eqref{Christoffel} that the parallel vector field $u$ along $\gamma$ is given by the relation 
		\begin{equation}\label{paralcond}
		\dot{u}^k-g_{ij}(x)u^i\dot{x}^jx^k=\dot{u}^k-g(u,\dot{\gamma})x^k=0,\quad k=1,\cdots,d,
		\end{equation}
		where $\phi(\gamma(s))=(x^1(s),\cdots,x^k(s))$ and 
		\[u(s)=\sum_{i=1}^{d} u^i(s)\partial_i=\sum_{i=1}^{d} u^i(s)d\phi^{-1}\left(\frac{\partial}{\partial x^i}\Big|_{\phi(\gamma(s))}\right).\]

		Now, let us denote the right-hand side of \eqref{parallel} as $P\bv:=P(\bv;\bp,\bq)$. Then, we can easily check $\langle P\bv,\bq\rangle_M =0$ from the following simple calculations:
		\[\begin{aligned}
		\langle P\bv,\bq\rangle_M =\langle\bv,\bq\rangle_M +\frac{\langle\bv,\bq\rangle_M }{1-\langle\bp,\bq\rangle_M }\left(\langle\bp,\bq\rangle_M +\langle\bq,\bq\rangle_M  \right)=\langle\bv,\bq\rangle_M -\langle\bv,\bq\rangle_M =0.
		\end{aligned} \] 
		On the other hand, the derivative of mapping
		\begin{equation*}
		s\mapsto P\left(\bv;\bp,\Gamma(s) \right)\in \mathbb{R}^{d+1}
		\end{equation*}
        can be obtained as 
		\begin{align*}
		\begin{aligned}
		&\frac{d}{ds}\left[\bm{v}+\frac{\langle \bm{v},\Gamma(s) \rangle_M }{1-\langle \bm{p}, \Gamma(s) \rangle_M }\big(\bm{p}+\Gamma(s)\big)\right]\\
		& \hspace{0.5cm} =~\frac{d}{ds}\left[\bm{v}+\frac{\langle \bm{v}, \dot{\Gamma}(0) \rangle_M \sinh(s)}{1+\cosh(s)}\left(\bm{p}+\cosh(s)\bm{p} + \sinh(s)\dot{\Gamma}(0)\right)\right]\\
		&\hspace{0.5cm} = ~\langle \bm{v}, \dot{\Gamma}(0) \rangle_M \cdot\frac{d}{ds}\left(\bm{p}\sinh(s)+(\cosh(s)-1)\dot{\Gamma}(0) \right)\\
		&\hspace{0.5cm} = ~\langle \bm{v}, \dot{\Gamma}(0) \rangle_M \cdot\left(\bm{p}\cosh(s)+\sinh(s)\dot{\Gamma}(0)\right)\\
		&\hspace{0.5cm} = ~\langle \bm{v}, \dot{\Gamma}(0) \rangle_M \cdot\Gamma(s),
		\end{aligned}
		\end{align*}
		where we used Proposition \ref{P2.2} and Proposition \ref{P2.4}. Moreover, the product 
		\[\left\langle P\left(\bv;\bp,\Gamma(s) \right), \dot{\Gamma}(s) \right\rangle_M  \]
		is conserved along the geodesic $\gamma$, since 
		\[\begin{aligned}
		&\frac{d}{ds}\left\langle P\left(\bv;\bp,\Gamma(s) \right), \dot{\Gamma}(s) \right\rangle_M \\
		&\hspace{0.5cm}=\left\langle \frac{d}{ds}P\left(\bv;\bp,\Gamma(s) \right), \dot{\Gamma}(s) \right\rangle_M +\left\langle P\left(\bv;\bp,\Gamma(s) \right), \ddot{\Gamma}(s) \right\rangle_M \\
		&\hspace{0.5cm}=\langle \bm{v}, \dot{\Gamma}(0) \rangle_M \cdot \left\langle \Gamma(s),\dot{\Gamma}(s)\right\rangle_M +\Big\langle P\left(\bv;\bp,\Gamma(s) \right), {\Gamma}(s)\Big\rangle_M \\
		&\hspace{0.5cm}=0,
		\end{aligned} \]
		where we used Proposition \ref{P2.2} to deduce $\ddot{\Gamma}(s)={\Gamma}(s)$ in the second equality. Therefore, the mapping 
		\begin{equation*}
	P\bv(s)=(p^0(s),\cdots,p^d(s)):=P\left(\bv;\bp,\Gamma(s) \right)
	\end{equation*}
	is a unique mapping satisfying 
	\[\sum_{i=0}^{d}p^0\frac{\partial}{\partial x^i}\Big|_{\iota(\gamma(s))}\in d\iota(T_{\gamma(s)}\mathbb{H}^d),\quad \frac{d}{ds}P\bv(s)=\langle P\bv(s), \dot{\Gamma}(s)\rangle_M  \cdot \Gamma(s), \]
	and the corresponding tangent vector $\displaystyle\sum_{i=1}^{d}p^i(s)d\phi^{-1}\left(\frac{\partial}{\partial x^i}\Big|_{\phi(\gamma(s))}\right)$ on $T_{\gamma(s)}\mathbb{H}^d$ satisfies the equation \eqref{paralcond} for $u$. Finally, we use the uniqueness of solution of \eqref{paralcond} to conclude $P\bv=\bu$.

\vspace{0.5cm}

\section{Proof of Lemma \ref{L4.4}} \label{App-C}
\setcounter{equation}{0}
In this appendix, we provide a complete proof for the two assertions in Lemma \ref{L4.4}.  \newline

Let $\Delta$ be a hyperbolic geodesic triangle on $\mathbb{H}^2$, and consider the hyperboloid model as a realization of $\mathbb{H}^2$. Recall that $A,B$ and $C$ are three vertices of $\Delta$, and  $b,c$ and angle $\angle A$ are defined as 
	\[b=d(C,A),~c=d(A,B),~\angle A: \emph{interior angle of}~\Delta ~\emph{at vertex}~A.  \]

\noindent $\bullet$~(Proof of the first assertion):~First, we introduce two well known facts on the hyperbolic trigonometry without proof:
\begin{enumerate}
	\item If the interior angle of $\Delta$ are given as $\angle A, \angle B$ and $\angle C$, the area of $\Delta$ can be written as
	\begin{equation}\label{area}
	\text{Area}(\Delta)=\pi-(\angle A+\angle B+\angle C),
	\end{equation}
which is a direct consequence of Gauss-Bonnet theorem.
	\item Let $a$ be the geodesic distance between $B$ and $C$. Then, there is a relation between $a,b,c$ and $\angle A$:
	\begin{equation}\label{cosine}
	\cosh a=\cosh b\cosh c-\sinh b\sinh c\cos \angle A,
	\end{equation}
	which is called as the hyperbolic law of cosines.
\end{enumerate}
\vspace{0.1cm}

		From now on, let us denote $s$ as the semiperimeter $s:=\frac{a+b+c}{2}$. Then, the law of cosines \eqref{cosine} can be reduced to the following two equivalent relations:
		\begin{equation*}
		\begin{aligned}
		\sin^2\Big(\frac{\angle A}{2}\Big)=&\frac{\sinh b\sinh c-\cosh b\cosh c+\cosh a}{2\sinh b\sinh c}=\frac{\sinh(s-c)\sinh(s-b)}{\sinh b\sinh c},
		\end{aligned}
		\end{equation*} and
		\begin{equation*}
		\cos^2\Big(\frac{\angle A}{2}\Big)=\frac{\sinh b\sinh c+\cosh b\cosh c+\cosh a}{2\sinh b\sinh c}=\frac{\sinh(s)\sinh(s-a)}{\sinh b\sinh c}.
		\end{equation*}
		
		\vspace{0.2cm}
		\noindent Then, these two relations yield
		\vspace{0.2cm}
		\begin{equation*}
		\begin{aligned}
		\cos\Big(\frac{\angle B+\angle C}{2}\Big)=&\sqrt{\frac{\sinh^2 s\sinh (s-b)\sinh (s-c)}{\sinh^2 a\sinh b\sinh c}}-\sqrt{\frac{\sinh^2 (s-a)\sinh (s-b)\sinh (s-c)}{\sinh^2 a\sinh b\sinh c}}\\
		=&\sin\Big(\frac{\angle A}{2}\Big)\frac{\sinh s}{\sinh a}-\sin\Big(\frac{\angle A}{2}\Big)\frac{\sinh (s-a)}{\sinh a}.
		\end{aligned}
		\end{equation*}
		Hence, we divide the above relation by $\sin\Big(\frac{\angle A}{2}\Big)$ do deduce
		\begin{equation*}
	\frac{\cos\Big(\frac{\angle B+\angle C}{2}\Big)}{\sin\Big(\frac{\angle A}{2}\Big)}=\frac{\sinh s-\sinh (s-a)}{\sinh a}=\frac{2\cosh\Big(\frac{b+c}{2}\Big)\sinh\Big(\frac{a}{2}\Big)}{\sinh a}=\frac{\cosh\Big(\frac{b+c}{2}\Big)}{\cosh(\frac{a}{2})},
		\end{equation*}
		and by using similar argument, one has
		\begin{equation*}
		\frac{\sin\Big(\frac{\angle B+\angle C}{2}\Big)}{\cos\Big(\frac{\angle A}{2}\Big)}=\frac{\cosh\Big(\frac{b-c}{2}\Big)}{\cosh(\frac{a}{2})}.
		\end{equation*}
		Therefore, we obtain
		\begin{equation}\label{tan}
		\tan\Big(\frac{\angle A}{2}\Big)\tan\Big(\frac{\angle B+\angle C}{2}\Big)=\frac{\cosh \Big(\frac{b-c}{2}\Big)}{\cosh \Big(\frac{b+c}{2}\Big)}=\frac{1-\tanh(\frac{b}{2})\tanh(\frac{c}{2})}{1+\tanh(\frac{b}{2})\tanh(\frac{c}{2})}.
		\end{equation}
		On the other hand, since the area of $\Delta$ can be represented by the interior angles, we have 
		\begin{equation*}
		\begin{aligned}
		\tan\left(\frac{Area(\Delta)}{2}\right) &=\cot\Big(\frac{\angle A+\angle B+\angle C}{2}\Big) =\frac{1-\tan\Big(\frac{\angle A}{2}\Big)\tan\Big(\frac{\angle B+\angle C}{2}\Big)}{\tan\Big(\frac{\angle A}{2}\Big)+\tan\Big(\frac{\angle B+\angle C}{2}\Big)}\\ &=\frac{\tan\Big(\frac{\angle A}{2}\Big)\Big(1-\tan\Big(\frac{\angle A}{2}\Big)\tan\Big(\frac{\angle B+\angle C}{2}\Big)\Big)}{\tan^2\Big(\frac{\angle A}{2}\Big)+\tan\Big(\frac{\angle A}{2}\Big)\tan\Big(\frac{\angle B+\angle C}{2}\Big)}.
		\end{aligned}
		\end{equation*}
		Therefore, as we substitute \eqref{tan} into the above relation, we get
		\begin{equation*}
		\begin{aligned}
		\tan\left(\frac{Area(\Delta)}{2}\right)&=\frac{2\tan(\frac{\angle A}{2})\tanh(\frac{b}{2})\tanh(\frac{c}{2})}{\tan^2\Big(\frac{\angle A}{2}\Big)+\tan^2\Big(\frac{\angle A}{2}\Big)\tanh(\frac{b}{2})\tanh(\frac{c}{2})+1-\tanh(\frac{b}{2})\tanh(\frac{c}{2})}\\
		&=\frac{2\tan(\frac{\angle A}{2})\tanh(\frac{b}{2})\tanh(\frac{c}{2})}{\sec^2\Big(\frac{\angle A}{2}\Big)+\tan^2\Big(\frac{\angle A}{2}\Big)\tanh(\frac{b}{2})\tanh(\frac{c}{2})-\tanh(\frac{b}{2})\tanh(\frac{c}{2})}\\
		&=\frac{\sin(\angle A)\tanh(\frac{b}{2})\tanh(\frac{c}{2})}{1+\sin^2\Big(\frac{\angle A}{2}\Big)\tanh(\frac{b}{2})\tanh(\frac{c}{2})-\cos^2\Big(\frac{\angle A}{2}\Big)\tanh(\frac{b}{2})\tanh(\frac{c}{2})}\\
		&=\frac{\sin(\angle A)\tanh(\frac{b}{2})\tanh(\frac{c}{2})}{1-\cos (\angle A)\tanh(\frac{b}{2})\tanh(\frac{c}{2})}=\frac{\sinh(\frac{b}{2})\sinh(\frac{c}{2})\sin (\angle A)}{\cosh(\frac{b}{2})\cosh(\frac{c}{2})-
			\sinh(\frac{b}{2})\sinh(\frac{c}{2})\cos (\angle A)},
		\end{aligned}
		\end{equation*}
		which is the desired result. \newline
		
		\noindent $\bullet$~(Proof of the second assertion) By taking left multiplication of proper Lorentz group $L\in O(1,2)$ if necessary, we may assume $LA=(1,0,0)$, $LB=(\cosh c,0,\sinh c)$ and  \[LC=(\cosh b, \sinh b\sin\angle A, \sinh b\cos\angle A). \]
		
%		Without loss of generality, we consider $\mathbb{H}^2$ and corresponding tensor matrix $(g_{ij})$:
%		\[ \mathbb{H}^2:=\{(x,y,z)|z^2-x^2-y^2=1\} \quad \mbox{and} \quad (g_{ij})=\mathrm{diag}[1,1,-1]. \]
%		Since $\mathbb{H}^2$ has maximal symmetry, it is possible to set a point $A$ in $\mathbb{H}^2$ as (0,0,1). By simple calculation, we have 
%		\begin{align*}
%		B=(\sinh(c),0,\cosh(c)), \qquad C=(\sinh(b)\cos(\angle A),\sinh(b)\sin(\angle A),\cosh(b))).
%		\end{align*}
\noindent Then, the determinant of matrix $\left(LA|LB|LC\right)$ satisfies
		\begin{align*}
		\Bigg|\det\begin{bmatrix}
		1&\cosh c&\cosh b\\0&0&\sinh b\sin(\angle A)\\0&\sinh c&\sinh b\cos(\angle A)
		\end{bmatrix}\Bigg|= \left|\sinh b\sinh c\sin \angle A\right|.
		\end{align*}
		Finally, since the determinant of $L$ is either $1$ or $-1$, we conclude the desired second assertion.

\bibliographystyle{amsplain}

\begin{thebibliography}{10}
	\bibitem{A-B} Acebron, J. A., Bonilla, L. L., P\'{e}rez Vicente, C. J. P., Ritort, F. and Spigler, R.: \textit{The Kuramoto model: A simple paradigm for synchronization phenomena.} Rev. Mod. Phys. {\bf 77} (2005), 137-185.
	
	\bibitem{A-H-S} Ahn, H., Ha, S.-Y. and Shim, W.: \textit{Emergent dynamics of a thermodynamic Cucker-Smale ensemble on complete Riemannian manifolds.} Submitted.
	
%	\bibitem{A-C-H-L} Ahn, S., Choi, H., Ha, S.-Y. and Lee, H.: \textit{On the collision avoiding initial-configurations to the Cucker-Smale type flocking models.} Commun. Math. Sci. {\bf 10} (2012), 625--643.
%	
%	\bibitem{A-H} Ahn, S. and Ha, S.-Y.: \textit{Stochastic flocking dynamics of the Cucker-Smale model with
%multiplicative white noises.} J. Math. Phys. {\bf 51} (2010), 103301.
%	
	\bibitem{A-B-F} Albi, G., Bellomo, N., Fermo, L., Ha, S.-Y., Kim, J., Pareschi, L., Poyato, D. and Soler, J.: \textit{Vehicular traffic, crowds, and swarms: from kinetic theory and multiscale methods to applications and research perspectives.} Math. Models Methods Appl. Sci. {\bf 29} (2019), 1901-2005.

	\bibitem{A-R} Aeyels,  D. and Rogge, J.: \textit{Stability of phase locking and existence of entrainment in networks of globally coupled oscillators.} Prog. Theor. Phys. {\bf112} (2004), 921-941.
	
	\bibitem{B-C-C} Ballerini, M., Cabibbo, N., Candelier, R., Cavagna, A., Cisbani, E., Giardina, I., Lecomte, V., Orlandi, A.,
Parisi, G., Procaccini, A., Viale, M. and Zdravkovic, V.: \textit{Interaction ruling animal collective behavior depends on topological rather than metric distance: evidence from a field study.} Proc. Natl. Acad. Sci. USA {\bf 105} (2008), 1232--1237.

	\bibitem{B} Barbălat, I.: \textit{Syst\`emes d\'equations diff\'erentielles d oscillations non Lin\'eaires,} Rev. Math. Pure Appl. {\bf 4} (1959), 267-270.
%	
	\bibitem{B-C-C} Bolley, F., Canizo, J. A. and Carrillo, J. A.: \textit{Stochastic mean-field
limit: non-Lipschitz forces and swarming.} Math. Mod. Meth. Appl. Sci. {\bf 21} (2011), 2179-2210.

	
	\bibitem{B-B} Buck, J. and  Buck, E.: \textit{Biology of synchronous flashing of fireflies}. Nature {\bf 211} (1966), 562-564.
	
	\bibitem{C-F-R-T} Carrillo, J. A., Fornasier, M., Rosado, J. and Toscani, G.: \textit{Asymptotic
flocking dynamics for the kinetic Cucker-Smale model.} SIAM J. Math. Anal.  {\bf 42} (2010), 218--236.

%\bibitem{C-D-H} Cho, H., Dong, J.-G. and Ha, S.-Y.: \textit{Emergent behaviors of a thermodynamic Cucker-Smale flock with a time-delay on a general digraph.} Submitted.
%
%\bibitem{C-H-H-J-K1} Cho, J., Ha, S.-Y., Huang, F., Jin, C. and Ko, D.: \textit{Emergence of bi-cluster flocking for the Cucker-Smale model.} Math. Models Methods Appl. Sci. {\bf 26} (2016), 1191-1218.
%
%\bibitem{CH17} Choi, Y.-P. and Haskovec, J.: \textit{Cucker-Smale model with normalized communication weights and time
%delay.} Kinetic and Related Models {\bf 10} (2017), 1011-1033.
%
%\bibitem{C-L} Choi, Y.-P. and Li, Z.: \textit{Emergent behavior of Cucker-Smale flocking particles with time lags.} Appl. Math. Lett. {\bf 86} (2018), 49-56.
%
\bibitem{C-H-L} Choi, Y.-P., Ha, S.-Y. and Li, Z.: \textit{Emergent dynamics of the Cucker-Smale flocking model and its variants.} In N. Bellomo, P. Degond, and E. Tadmor (Eds.), Active Particles Vol.I - Theory, Models, Applications(tentative title), Series: Modeling and Simulation in Science and Technology, Birkhauser-Springer.

%\bibitem{C-P} Choi, Y.-P. and Pignotti, C.: \textit{Emergent behavior of Cucker-Smale model with normalized weights and distributed time delays.} Netw. Heterog. Media {\bf 14} (2019), 789-804.
%
%\bibitem{C-S}  Choi, Y.-P. and Salem, S.: \textit{Cucker-Smale flocking particles with multiplicative noises: stochastic mean-field limit and phase transition.} Kinet. Relat. Models {\bf 12} (2019), 573-592.
%	
	\bibitem{C-D} Cucker, F. and Dong, J.-G.: \textit{Avoiding collisions in flocks.} IEEE Trans. Automatic Control {\bf 55}, 1238--1243 (2010).
	
	\bibitem{C-D1} Cucker, F. and Dong, J.-G.: \textit{On the critical exponent for flocks under hierarchical leadership.} Math. Models Methods Appl. Sci. {\bf 19} (2009), 1391-1404.

\bibitem{C-S1} Cucker, F. and Smale, S.: \textit{On the mathematics of emergence}. Japan. J. Math. {\bf 2}, 197--227  (2007).
	
\bibitem{C-S2} Cucker, F. and Smale, S.: \textit{Emergent behavior in flocks}. IEEE Trans. Automat. Control {\bf 52}, 852--862 (2007).
	
	\bibitem{D-M1} Degond, P. and Motsch, S.: \textit{Macroscopic limit of self-driven particles with orientation interaction.}  C.R. Math. Acad. Sci. Paris {\bf 345} (2007), 555-560.

\bibitem{D-M2} Degond, P. and Motsch, S.: \textit{Large-scale dynamics of the persistent Turing Walker model of fish behavior.} J. Stat. Phys. {\bf 131} (2008), 989-1022.

\bibitem{D-M3} Degond, P. and Motsch, S.: \textit{Continuum limit of self-driven particles with orientation interaction.} Math. Mod. Meth. Appl. Sci. {\bf 18} (2008), 1193-1215.

\bibitem{DS2019} Dietert H., Shvydkoy, R.: \textit{On Cucker-Smale dynamical systems with degenerate communication.} {Analysis and Applications.} (2020) (DOI) https://doi.org/10.1142/S0219530520500050.
%Arch. Rational Mech. Anal. (2019) (DOI) https://doi.org/10.1007/s00205-019-01452-y.

%\bibitem{D-H-K1} Dong, J.-G., Ha, S.-Y. and Kim, D.: \textit{On the Cucker?Smale ensemble with the q-closest neighbors in a self-consistent temperature field}. 
% SIAM Journal on Control and Optimization {\bf 58} (2020) 368-392.
%
%\bibitem{D-H-K2} Dong, J.-G., Ha, S.-Y. and Kim, D.: \textit{On the Cucker-Smale with q-closest neighbors under time-delayed communications.} Kinetic and Related Models {\bf 13} (2020) 653-676.
%
%\bibitem{D-H-J-K1} Dong, J.-G., Ha, S.-Y., Jung, J. and Kim, D.: \textit{On the stochastic flocking of the Cucker-Smale flock with randomly switching topologies.} To appear in SIAM Journal on Control and Optimization.
%
%\bibitem{D-H-J-K2} Dong, J.-G., Ha, S.-Y., Jung, J. and Kim, D.: \textit{ Emergence of stochastic flocking for the discrete Cucker-Smale model with randomly switching topologies}. To appear. 

\bibitem{D-Q} Dong, J.-G. and Qiu, L.: \textit{Flocking of the Cucker-Smale model on general digraphs.} {\it IEEE Trans. Automat. Control} {\bf 62} (2017), 5234-5239.

\bibitem{D-F-T} R. Duan, M. Fornasier and G. Toscani, A kinetic flocking model with diffusion, {\it Commun. Math. Phys.} {\bf 300} (2010), 95--145.

	\bibitem{D-B1} D\"{o}rfler, F. and Bullo, F.: \textit{Synchronization in complex networks of phase oscillators: A survey. } Automatica {\bf 50} (2014), 1539-1564.
	
%	\bibitem{E-H-S} Erban, R., Haskovec, J. and Sun, Y.: \textit{On Cucker-Smale model with noise and delay.} SIAM. J. Appl. Math. {\bf 76} (2016), 1535-1557.
%
%	
%	\bibitem{H-H-K} Ha, S.-Y. Ha, T. and Kim, J.:  \textit{Asymptotic flocking dynamics for the Cucker-Smale model with the Rayleigh friction.} {\it J. Phys. A: Math. Theor.} {\bf 43} (2010), 315201.
	
	\bibitem{H-H-K-K-M}  Ha, S.-Y., Hwang, S., Kim, D., Kim, S.-C. and Min, C.: \textit{Emergent behaviors of a first-order particle swarm model on the hyperboloid}. J. Math. Phys. {\bf 61} (2020), 042701.
	
%	\bibitem{H-J-J} Ha, S.-Y., Jin, S. and Jung, J.: \textit{A local sensitivity analysis for the kinetic Cucker-Smale equation with random inputs.} J. Differ. Equations {\bf 265} (2018) 3618-3649.
%	
%	\bibitem{H-J-J-S} Ha, S.-Y., Jin, S., Jung, J. and Shim, W.: \textit{A local sensitivity analysis for the hydrodynamic Cucker-Smale equation with random inputs.} J. Differ. Equations {\bf 268} (2020), 636-679.
%	
%	\bibitem{H-J-K} Ha, S.-Y., Jeong, E. and Kang, M.-J.: \textit{Emergent behavior of a generalized Viscek-type flocking model.} Nonlinearity {\bf 23} (2010),  3139-3156.
%		
%	\bibitem{H-J-K-P-Z} Ha, S.-Y., Jung, J., Kim, J., Park, J., and Zhang, X.: \textit{Emergent behaviors of the swarmalator model for position-phase aggregation.} Math. Models Methods Appl. Sci {\bf 29} (2019), 2225-2269.
	
	\bibitem{H-K-P-Z} Ha, S.-Y., Kim, J., Park, J. and Zhang, X.: \textit{Complete cluster predictability of the Cucker-Smale flocking model on the real line.} Arch. Ration. Mech. Anal. {\bf 31} (2019), 319-365.
	
	\bibitem{H-K-Rug} Ha, S.-Y.,  Kim, J., Ruggeri,T.: \textit{Emergent behaviors of thermodynamic Cucker-Smale
		particles}. SIAM J. Math. Anal. {\bf 50} (2019),  3092–3121.
	
	\bibitem{H-K-Rug2}  Ha, S.-Y., Kim, J., Ruggeri, T.:
	\textit{From the Relativistic Mixture of Gases
		to the Relativistic Cucker–Smale Flocking}.  
	%Arch. Rational Mech. Anal. (2019) (DOI) https://doi.org/10.1007/s00205-019-01452-y.
	Arch. Ration. Mech. Anal. {\bf 235} (2020), 1661–1706. 
	
	\bibitem{H-K-R} Ha, S.-Y., Ko, D. and Ryoo, S. W.: \textit{Emergent dynamics of a generalized Lohe model on some class of Lie groups}. J. Stat. Phys. {\bf 168} (2017), 171-207.
%	\bibitem{H-K-R} Ha, S.-Y., Kim, H. W. and Ryoo, S. W.: \textit{Emergence of phase-locked states for the Kuramoto model in a large coupling regime}. Commun. Math. Sci.  {\bf 14} (2016), 1073-1091.
	
	\bibitem{H-K-S} Ha, S.-Y., Kim, D. and Schl\"oder, F.W.: \textit{Emergent behaviors of Cucker-Smale flocks on Riemannian manifolds}. Preprint.
		
%		\bibitem{H-L-L} Ha, S.-Y., Lee, K. and Levy, D.: \textit{Emergence of time-asymptotic flocking in a stochastic
%Cucker-Smale system.} Commun. Math. Sci. {\bf 7} (2009), 453-469.

\bibitem{H-Liu} Ha, S.-Y. and Liu, J.-G.: \textit{A simple proof of Cucker-Smale flocking dynamics and mean field limit.} Commun. Math. Sci. {\bf 7} (2009), 297-325.
	
	\bibitem{H-R} Ha, S.-Y., Ruggeri,T.: \textit{Emergent dynamics of a thermodynamically consistent particle model}. Arch. Ration. Mech. Anal. {\bf 223} (2017), 1397-1425.
	
%	\bibitem{H-Ryoo} Ha, S.-Y. and Ryoo, S. W.: \textit{Asymptotic phase-locking dynamics and critical coupling strength for the Kuramoto model}. Submitted.
%	
	\bibitem{H-T} Ha, S.-Y. and Tadmor, E.: \textit{From particle to kinetic and hydrodynamic description of flocking.} Kinetic Relat. Models {\bf 1} (2008), 415-435.

	\bibitem{Ji} Jin, C.: \textit{Flocking of the Motsch-Tadmor model with a cut-off interaction function.} J. Stat. Phys. {\bf 171} (2018), 345-360.

	\bibitem{Ku2} Kuramoto, Y.: \textit{International symposium on mathematical problems in mathematical physics}. Lecture Notes Theor. Phys. {\bf 30} (1975), 420.
	
	\bibitem{L-X}  Li, Z. and Xue, X.: \textit{Cucker-Smale flocking under rooted leadership with fixed and switching topologies.} SIAM J. Appl. Math. {\bf 70} (2010), 3156--3174.

\bibitem{M-T} Motsch, S. and Tadmor, E.: \textit{A new model for self-organized dynamics and its flocking behavior.} J. Stat. Phys. {\bf 144} (2011), 923--947.

\bibitem{MT2014} Mostch, S. and Tadmor, E.: \textit{Heterophilious dynamics enhances consensus.} SIAM REV. {\bf 56} (2014), 577-621.
	
	\bibitem{Pe} Peskin, C. S.: \textit{Mathematical aspects of heart physiology}. Courant Institute of Mathematical Sciences, New York, 1975.
	
	\bibitem{PV17} Pignotti, C. and Vallejo, I. R.: \textit{Flocking estimates for the Cucker-Smale model with time lag and hierarchical leadership.} 
 J. Math. Anal. Appl. {\bf 464} (2018), 1313-1332. 

\bibitem{P-S} Poyato, D and Soler, J.: \textit{Euler-type equations and commutators in singular and hyperbolic limits of kinetic Cucker-Smale models.} Math. Mod. Meth. Appl. Sci. {\bf 6} (2017), 1089-1152.

	\bibitem{P-R} Pikovsky, A., Rosenblum, M. and Kurths, J.: \textit{Synchronization: A universal concept in
		nonlinear sciences}. Cambridge University Press, Cambridge, 2001.

	\bibitem{R-L-W} Ritchie, L.M., Lohe, M.A. and Williams, A.G.: \textit{Synchronization of relativistic particles in the hyperbolic Kuramoto model.} Chaos. {\bf  28} (2018), 053116.
	
	\bibitem{St} Strogatz, S. H.: \textit{From Kuramoto to Crawford: exploring the onset of synchronization in populations of coupled oscillators.} Phys. D {\bf 143} (2000), 1-20.
	
	\bibitem{T-T} Toner, J. and Tu, Y.: \textit{Flocks, herds, and schools: A quantitative theory of flocking.} {\it Phys. Rev. E} {\bf 58} (1998), 4828-4858.

    \bibitem{T-B} Topaz, C. M. and Bertozzi, A. L.: \textit{Swarming patterns in a two-dimensional kinematic model for biological groups.} {\it SIAM J. Appl. Math.} {\bf 65} (2004), 152-174.

	\bibitem{Wi2} Winfree, A. T.: \textit{Biological rhythms and the behavior of populations of coupled oscillators.} J. Theor. Biol. {\bf 16} (1967), 15-42.
	
	\bibitem{Wi1} Winfree, A. T.: \textit{The geometry of biological time}. Springer, New York, 1980.
	
	
\end{thebibliography}

\end{document}